\documentclass[a4paper,12pt]{article}

\usepackage{a4wide}
\usepackage[english]{babel}
\usepackage{geometry}
     \usepackage{setspace}
\usepackage{amsfonts}
\usepackage{amsmath}
\usepackage{graphicx}
\usepackage{caption}
\usepackage[colorlinks,bookmarks=true]{hyperref}
\usepackage{amsthm}
\usepackage{enumitem}

\newtheorem{lemma}{Lemma}[section]
\newtheorem{definition}{Definition}[section]
\newtheorem{theorem}{Theorem}[section]
\newtheorem{prop}{Proposition}[section]
\newtheorem{remark}{Remark}[section]

\newcommand{\E}{\mathbb{E}}
\def\Ac{{\cal A}}
\def\Cc{{\cal C}}

\def\Sc{{\cal S}}

\def\Zc{{\cal Z}}

\def\Zc{{\cal Z}}

\def\eps{\varepsilon}

\def\05{\frac{1}{2}}
\def\-1{^{-1}}
\def\1{{1\hspace{-1mm}{\rm I}}}
\def\={\;=\;}
\def\.{\;.}

\title{  }
\author{ }

\def \E{\mathbb{E}}
\def \F{\mathbb{F}}

\def \N{\mathbb{N}}
\def \R{\mathbb{R}}

\def\P{\mathbb{P}}

\title{Optimal make-take fees for market making regulation}
\author{Omar El Euch, Thibaut Mastrolia, Mathieu Rosenbaum, Nizar Touzi\\ \'Ecole Polytechnique
\thanks{This work benefits from the financial support of the Chaires {\it Analytics and Models for Regulation}, {\it Financial Risk} and {\it Finance and Sustainable Development}. Omar El Euch and Mathieu Rosenbaum gratefully acknowledge the financial support of the {\it ERC Grant 679836 Staqamof}.
Nizar Touzi gratefully acknowledges the financial support of the {\it ERC Grant 321111 RoFiRM.}
Thibaut Mastrolia gratefully acknowledges the financial support of the {\it ANR PACMAN.} The authors also thank Bas Werker for relevant comments.
}}

\begin{document}

\maketitle

\begin{abstract}
\noindent We address the mechanism design problem of an exchange setting suitable make-take fees to attract liquidity on its platform.
Using a principal-agent approach, we provide the optimal compensation scheme of a market maker in quasi-explicit form. This contract depends essentially on
the market maker inventory trajectory and on the volatility of the asset. We also provide the optimal quotes that
should be displayed by the market maker. The simplicity of our formulas allows us to analyze in details
the effects of optimal contracting with an exchange, compared to a situation without contract.
We show in particular that it improves liquidity and reduces trading costs for investors. We extend our study to an oligopoly of symmetric exchanges and we study the impact of such common agency policy on the system.
\end{abstract}

\noindent \textbf{Keywords:} Make-take fees, market making, financial regulation, high-frequency trading, principal-agent problem, stochastic control.

\section{Introduction}\label{intro}
Due to the fragmentation of financial markets, exchanges are nowadays in competition. The traditional international exchanges are challenged by alternative trading venues, see \cite{lehalle2013market}. Consequently,
they have to find innovative ways to attract liquidity on their platforms. One solution is to use a
maker-taker fees system, that is a rule enabling them to charge in an asymmetric way
liquidity provision and liquidity consumption. The most classical setting, used by many exchanges 
(such as Nasdaq, Euronext, BATS Chi-X...), is of course to subsidize the former while taxing the latter. 
In practice, this means associating a fee rebate to executed limit orders and applying 
a transaction cost for market orders.\\

\noindent In the recent years, the topic of make-take fees has been quite controversial. 
Make-take fees policies are seen as a major facilitating factor 
to the emergence of a new type of market makers aiming at collecting fee rebates: the high frequency traders. 
As stated by the Securities and Exchanges commission in \cite{sec}:
``Highly automated exchange systems and liquidity rebates have helped establish a business model for a new type of professional liquidity provider that is distinct from the more traditional exchange specialist and over-the-counter market maker."
The concern with high frequency traders becoming the new liquidity providers is two-fold. 
First, their presence implies that slower traders no longer have access to the limit order book, or only in unfavorable 
situations when high frequency traders do not
wish to support liquidity. This leads to the second classical criticism against high frequency market makers: 
they tend to leave the market in time of stress, see \cite{bellia2017high, pamela, menkveld2013high, foresight} for detailed investigations about high frequency market making activity.\\

\noindent From an academic viewpoint, studies of make-take fees structures and their impact on the welfare of the markets 
have been mostly empirical, or carried out in rather stylized models. An interesting theory, suggested in \cite{angel2011equity} and developed in \cite{colliard2012trading} is that make-take fees have actually no impact on trading
costs in the sense that the {\it cum fee} bid-ask spread should not depend on the make-take fees policy.
This result is consistent with the empirical findings in \cite{lutat2010effect,malinova2015subsidizing}. Nevertheless, it is clearly shown
in these works that many important trading parameters such as depths, volumes or price impact do depend on the make-take fees structure, see also \cite{harris2013maker}.
Furthermore, the idea of the neutrality of the make-take fees schedule is also tempered in \cite{foucault2013liquidity} where the authors show theoretically that make-take fees may increase welfare of markets
provided the tick size is not equal to zero, see also \cite{brolley2013informed}. More importantly, the results above are obtained in tractable but rather simple discrete-time models that one may want to revisit to be closer to market reality.\\

\noindent In this work, our goal is to provide a quantitative and operational answer to the question of relevant make-take fees.
To do so, we take the position of an exchange (or of the regulator) aiming at attracting liquidity. The exchange is
looking for the best make-take fees policy to offer to market makers in order to maximize its utility. In other words, it wants to
design an optimal contract with the market marker to create an incentive to increase liquidity. We solve the problem the exchange is facing and we do not consider the more involved question of global social welfare. Nevertheless, we have in mind that increasing the liquidity (what the exchange is aiming at) should be beneficial for the welfare of all the agents, what is confirmed in our empirical results. This paper is to our knowledge the first addressing the issue of make-take fees in a realistic continuous-time framework. As a first step, we consider a single market maker in a non-fragmented market, such as for example many fixed income markets which represent some of the most liquid assets in the world. We next consider the case of multiple symmetric exchanges.\\

\noindent Incentive theory has emerged in the 1970s in economics to model how a financial agent can delegate the management of an output process to another agent. Let us recall the formalism of principal-agent problems from the seminal works of Mirrlees \cite{mirrlees1974notes} and Holmstr\"om \cite{holmstrom1979moral}. A principal aims at contracting with an agent who provides efforts to manage an output process impacting the wealth of the principal. The principal is not able to control directly the output process since she cannot decide the efforts made by the agent. In our case, the principal is the exchange, the agent is the market maker, the effort corresponds to the quality of the liquidity provided by the market maker (essentially the size of the bid-ask spread proposed by the market maker), the output process is the transactions flow on the platform and the contract depends on the realized transactions flow. Several economics papers have investigated this kind of problems by identifying it with a Stackelberg equilibrium between the two parties. More precisely, since the principal cannot control the work of the agent, she anticipates his best response to a given compensation. We follow the stream of literature initiated in \cite{holmstrom1987aggregation}. Then in \cite{sannikov2008continuous}, the author recasts such issue into a stochastic control problem which has been further developed using backward stochastic differential equation theory in \cite{cvitanic2015dynamic}. See also \cite{cvitanic2013contract} for related literature.\\

\noindent This paper provides a quasi-explicit expression for the optimal contract between the exchange and the market maker, and for the market maker optimal quotes. The optimal contract depends essentially on the market maker inventory trajectory and on the volatility of the market. These simple formulas enable us to analyze in details the effects for the welfare of the market of optimal contracting with an exchange, compared to a situation without contract as in \cite{avellaneda2008high, gueant2013dealing}. We show that such contracts lead to reduced spreads and lower trading costs for investors. We also propose an extension of this work to an oligopoly of symmetric exchanges aiming et hiring a single market marker. We show in particular that there exists a unique Markovian symmetric Nash equilibrium for the exchanges.\\

\noindent The paper is organized as follows. Our modeling approach is presented in Section \ref{model}. In particular, we define the market maker's as well as the exchange's optimization {framework}. In Section \ref{solvingmmp}, we compute the {best response} of the market maker for a given contract. Optimal contracts are designed in Section \ref{opticont} where we solve the exchange's problem. Then, in Section \ref{section:comparison}, we assess the benefits for market quality of the presence of an exchange contracting optimally with a market maker. In Section \ref{section:oligo}, we extend our study to an oligopoly of symmetric exchanges. Finally, useful technical results are gathered in an appendix together with the so-called first best case (see Appendix \ref{sec:firstbest}) which provides different solutions.

\section{The model}\label{model}

Our starting point is the seminal work of Avellaneda \& Stoikov \cite{avellaneda2008high}. Our objective is to derive optimal make-take fees in order to monitor the behavior of a market maker on a platform acting according to the optimal market making model of \cite{avellaneda2008high}.

\subsection{Contractible and observable variables}

Let $T>0$ be a final horizon time, $\Omega_c $ the set of continuous functions from $[0,T]$ into $\mathbb R$, $\Omega_d$ the set of piecewise constant c\`adl\`ag functions from $[0,T]$ into $\mathbb N$, and $\Omega = \Omega_c \times (\Omega_d)^2$ with corresponding Borel algebra $\mathcal F$. The observable state is the canonical process $(\chi_t)_{t \in [0,T]} = (S_t, N^a_t, N^b_t)_{t \in [0,T]}$ of the measurable space $(\Omega, \cal F)$: 
 \begin{eqnarray*}
 S_t(\omega) := s(t),~~ N^a_t(\omega) := n^a(t),~~ N^b_t(\omega) := n^b(t),
 &\mbox{for all}&
 t\in[0,T],~ \omega = (s,n^a, n^b) \in \Omega,
 \end{eqnarray*}
with canonical completed filtration $\F = ({\cal F}_t)_{t \in [0,T]} =  ({\cal F}_t^c \otimes ({\cal F}_t^d)^{\otimes2} )_{t \in [0,T]}$.
 
The trading activity is reduced to a single risky asset $S$ with observable efficient price $S$ defined by:\footnote{In practice, the efficient price can be thought of as the mid-price of the asset, see \cite{robert2010new,delattre2013estimating}.}
\begin{equation} \label{spotDynamics} 
S_t := S_0 +\sigma W_t,  \quad t \in [0,T],
\end{equation}
for some Brownian motion $W$, initial price $S_0>0$, and constant volatility $\sigma>0$. The market maker chooses processes denoted by $\delta^b$ and $\delta^a$ respectively so as to fix publicly available bid and ask offer prices:
 \begin{eqnarray*} 
 P^b_t := S_t - \delta^b_t
 &\mbox{and}&
 P^a_t := S_t + \delta^a_t,
 ~~t\in[0,T].
 \end{eqnarray*}
The arrival of bid and ask market orders is modeled by a counting process $(N^b,N^a)$ with unit jumps, so that no more than one market order can occur at each time. We introduce the inventory process of the market maker $Q$:
 \begin{eqnarray*}
 Q_t 
 := 
 N^b_t - N^a_t
 \;\in\; \mathbb N\cap\big[-\bar q,\bar q \big],
 &t \in [0, T],& 
 \end{eqnarray*}
where $N^b_0=N^a_0=0$ and, as in \cite{gueant2013dealing}, we impose a critical absolute inventory $\bar q \in \N$ above which the market maker stops quoting on the ask or bid side. 

Let $c>0$ be the fee collected by the exchange, see Section \ref{coutTransaction} below. In order to illustrate the impact of the posted prices on the transactions arrival process $(N^b,N^a)$, the corresponding intensity process depends on the departure of the transaction price from the efficient price, i.e. $c+\delta^i_t$, $i\in\{b,a\}$, as follows:
 \begin{eqnarray} \label{intensity} 
 \lambda^{i,\delta}_t 
 \;:=\; 
 \lambda(\delta^i_t) \1_{\{ \varepsilon_iQ_t >-  \bar q\}}, 
 &i\in\{b,a\},~~(\varepsilon_b,\varepsilon_a)=(-1,1),& 
 \mbox{with}~~
 \lambda(x) := A e^{-k \frac{(x+c)}{\sigma}},
 \end{eqnarray}
for some fixed positive constants $A$ and $k$, and with dependence on the volatility parameter $\sigma$ so as to reproduce the stylized fact that the average number of trades per unit of time is a decreasing function of the ratio between the spread and the volatility, see \cite{dayri2015large, madhavan1997security, wyart2008relation}. \\

Our canonical variables being $S$, $N^a$ and $N^b$, the contracts are allowed to depend on the trajectories of these quantities only: these are our three contractible variables. This is actually very reasonable: the efficient price is a quantity any market participant is used to, whether the chosen proxy for it is the midprice, the last traded price or some volume weighted
price. The processes $N^a$ and $N^b$ encode the arrivals of market orders and therefore actual transactions, which are clearly recorded on any exchange and accessible to most participants.
So the contracts will be designed on standard, unarguable and easily obtained financial variables.\\

Note that the spreads $\delta^a$ and $\delta^b$ are here observable by market takers, but not contractible. From the exchange viewpoint, it would not be reasonable
to introduce the spread variable in a contract. First, quotes are typically not recorded with the same degree of accuracy as transactions since they are evolving on a much higher frequency, which can sometimes be of the order of the 
millisecond. Second, quoted prices do not in general lead to transactions and it would be probably hard to justify taxing or subsidizing agents based on offer prices having no tangible 
counterparts. Finally, a quote based contract would certainly encourage high frequency traders to attempt exploiting possible flaws in the contract, using for example a very high rate
of cancellations of their orders, leading to possible market disruptions, see \cite{abergel2014understanding,pamela} for studies about high frequency traders behavior.

\subsection{Admissible controls and market maker's problem}

The set of admissible controls ${\cal A} $ is the collection of all predictable processes $\delta=(\delta^a,\delta^b)$ uniformly bounded by $\delta_\infty$, some sufficiently large positive constant to be fixed later. Controlling $\delta$ is equivalent to control the arrival intensity of market orders since it is a deterministic function of the spread. Viewing the market maker optimization problem this way, we see that the intensity
plays the very same role
as the drift in standard principal-agent problems where the agent controls the drift of a diffusion process, this drift being unobserved by the principal, see \cite{sannikov2008continuous, cvitanic2015dynamic}. 
A particular feature of our modelling is that the intensity is observable (although non-contractible) because of its connexion with the spread. However, the spread is in some sense an artificial variable here enabling us to fit with market reality (the market maker has access in practice to the spread but not to the intensity of market takers arrivals). Each control process $\delta$ induces 

-- the market maker profit and loss process:
 \begin{eqnarray} \label{PL} 
 \mbox{PL}_t^{\delta} 
 \;:=\; 
 X_t^{\delta} + Q_t S_t ,
 &\mbox{where}&
 X_t^{\delta} \;:=\; \int_0^t P^a_r dN^a_r - \int_0^t P^b_r dN^b_r,
 ~~t\in [0,T],
 \end{eqnarray}
as the sum of the cash flow $X^\delta$ and the inventory risk\footnote{As in \cite{avellaneda2008high}, for sake of simplicity, we assume that the market maker estimates his inventory risk using the efficient price $S$.} $QS$,

-- and a probability measure $\P^{\delta}$ under which $S$ is driven by \eqref{spotDynamics}, and
 \begin{eqnarray*}
 \widetilde N^{i,\delta}_t := N^i_t - \int_0^t \lambda^{i,\delta}_rdr, 
 &  
 t\in[0,T],~
 i\in\{b,a\},
 &
 \mbox{are martingales.}
 \end{eqnarray*}
Then, $\P^\delta$ is defined by the density $\frac{d\P^\delta}{d\mathbb P^0}\big|_{{\cal F}_T}
 = 
 L^\delta_T$, induced by the Dol\'eans-Dade exponential martingale 
 $dL^\delta_t= L^\delta_{t-}\sum_{i=b,a}\1_{\{\eps_iQ_{t-}>-\bar q\}}\frac{\lambda(\delta^i_t ) - A}{A} d\widetilde N^{i,0}_t$, i.e.\footnote{see e.g. Theorem III.3.11 in \cite{jacod2013limit}; the uniform boundedness of $\delta$ guarantees that $L^\delta$ is a martingale, see \cite{sokol2013optimal}.} again with $(\varepsilon_b,\varepsilon_a)=(-1,1)$,
 \begin{eqnarray}  \label{proba_change}
 L^\delta_t 
 &:=&  
 \exp{\sum_{i=b,a}\int_0^t \1_{\{\eps_iQ_{r-}>-\bar q\}}
                                      \Big[ \log\Big(\frac{\lambda(\delta^i_r)}{A}\Big)  dN^i_r 
                                              - (\lambda(\delta^i_r) - A) dr
                                      \Big] },
 ~~t\in[0,T].
 \end{eqnarray}
In particular, all probability measures $\mathbb P^{\delta}$ are equivalent. We therefore use the notation $a.s$ for almost surely without ambiguity.  We shall write $\mathbb E_t^\delta$ for the conditional expectation with respect to $\mathcal F_t$ with probability measure $\mathbb P^\delta$.

The exchange aims at encouraging the market maker to reduce his spread so as to enhance market liquidity on the platform. This is achieved by setting the terms of an incentive contract defined by an $\mathcal F_T$-measurable random variable $\xi$. In other words, the compensation $\xi$ may depend on the whole paths of the contractible variables $N^a$, $N^b$ and $S$. Given this additional revenue, the market maker's objective is defined by the utility maximization problem
 \begin{eqnarray} \label{MMpb}
 V_{\text{MM}}(\xi) 
 \;:=\; 
 \sup_{\delta \in {\cal A}} 
 J_{\text{MM}}(\delta,\xi),
 ~~\mbox{where}~~
 J_{\text{MM}}(\delta,\xi)
 &:=&
 \E^{\delta}\Big[ - e^{-\gamma (\xi+\text{PL}^{\delta}_T)} \Big]
 \\
 &=&
 \E^{\delta}\Big[ - e^{-\gamma (\xi+\int_0^T \delta^a_tdN^a_t+\delta^b_tdN^b_t+Q_tdS_t)} \Big].
 \nonumber
 \end{eqnarray}
Here, $\gamma>0$ is the absolute risk aversion parameter of the CARA market maker.  For each compensation $\xi$, we shall prove below that there exists a unique optimal response $\hat\delta(\xi) = (\hat\delta^b(\xi), \hat\delta^a(\xi))\in\cal A$ of the market marker, i.e. $V_{\text{MM}}(\xi) 
 =J_{\text{MM}}\big(\hat\delta(\xi),\xi\big)$.

\begin{remark} When there is no incentive payment $\xi = 0$, the utility maximization problem \eqref{MMpb} reduces to the Avellaneda \& Stoikov \cite{avellaneda2008high, gueant2013dealing} optimal market making problem. 
\end{remark}

\subsection{The exchange optimal contracting problem}\label{coutTransaction}

The exchange receives a fixed fee $c > 0$ for each market order that occurs in the market\footnote{In practice, some exchanges add to this fixed fee a component which is proportional to the traded cash amount. Our analysis can be extended to more elaborated fee schedules. Our choice of a constant fee is motivated by the induced simplicity which will be crucial to derive our quasi-explicit solution. Furthermore, we will in fact see that when using the optimal contract, the exchange is somehow indifferent to the value of $c$, see Section \ref{discussion}.},
and then collects at time $T$ the total revenue $c(N^a_T+ N^b_T) - \xi$. The choice of the contract $\xi$ is dictated by the utility maximization problem 
 \begin{eqnarray} \label{Ppb} 
 V^E_0 
 &:=& 
 \sup_{\xi \in {\cal C}} ~ \E^{\hat\delta(\xi)} 
                                     \Big[-e^{-\eta(c(N^a_T+N^b_T) - \xi)}\Big],
\end{eqnarray}
where $\eta>0$ is the exchange's absolute risk aversion parameter, and the set of admissible contracts $\mathcal C$ is the collection of all contracts satisfying 

-- the participation constraint $V_{\text{MM}}(\xi) \geq R$, where the reservation level $R < 0$ may be chosen to be the utility level without contract,

-- together with the integrability conditions:
\begin{eqnarray}\label{condition:xi}
\sup_{\delta\in\cal A} \mathbb E^{\delta} \Big[e^{\eta' \xi}\Big] 
<\infty
~~\mbox{and}~~
\sup_{\delta\in\cal A} \mathbb E^{\delta} \Big[e^{- \gamma' \xi}\Big] 
<\infty,
&\mbox{for some}&
\eta'>\eta,~~\gamma'>\gamma.
\end{eqnarray}
Since $N^a$ and $N^b$ are point processes with bounded intensities, and the inventory process $Q$ is bounded, it follows from an easy application of the H\"older inequality that  the expectations in both problems \eqref{MMpb} and \eqref{Ppb} are finite. 

We assume throughout this paper that the participation level $R$ is so that the set of admissible contracts is non-empty:
 \begin{eqnarray*}
 \cal C 
 &=& 
 \Big\{ \xi,\, \text{ ${\cal F}_T$-measurable such that $V_{\text{MM}}(\xi) \geq R$ and \eqref{condition:xi} is satisfied} 
 \Big\}
 \;\neq\;\emptyset.
 \end{eqnarray*}

\section{Solving the market maker's problem}\label{solvingmmp}

We start by solving the problem \eqref{MMpb} of the market maker facing an arbitrary contract $\xi \in \Cc$ proposed by the exchange.

\subsection{Market maker's optimal response}

For $(\delta,z,q)\in[-\delta_\infty,\delta_\infty]^2\times \mathbb R^3\times \mathbb Z$, with $\delta=(\delta^a,\delta^b)$ and $z=(z^S,z^a,z^b)$, we define 
 \begin{eqnarray*}
 h(\delta, z, q)
 := 
 \sum_{i=b,a}
 \frac{1-e^{-\gamma(z^i + \delta^i)}}{\gamma}\,\lambda(\delta^i) \1_{\{\varepsilon_iq > - \bar q\}}
 &\mbox{and}&
 H(z, q)
 := 
 \sup_{|\delta^a|\lor|\delta^b| \leq  \delta_\infty}  h(\delta, z, q).
 \end{eqnarray*}
with $(\varepsilon_b,\varepsilon_a)=(-1,1)$. 
For an arbitrary constant $Y_0\in\R$ and predictable processes $Z=(Z^S, Z^a, Z^b)$, {with $\int_0^T \big(|Z^S_t|^2 + |H(Z_t, Q_t)|\big)dt<\infty$},  we introduce the process
 \begin{eqnarray}\label{YZ}
 Y^{Y_0,Z}_t
 \;=\;
 Y_0 + \int_0^t Z^a_r dN^a_r + Z^b_r dN^b_r + Z^S_r dS_r 
                         + \Big(\frac12\gamma \sigma^2 (Z^S_r + Q_r)^2-H(Z_r, Q_r)\Big)dr,
 \end{eqnarray}
and we denote by $\mathcal Z$ the collection of all such processes $Z$ such that the first integrability condition in \eqref{condition:xi} is satisfied with $\xi=Y_T^{0,Z}$ and 
\begin{equation} \label{conditionY}
\sup_{\delta \in \Ac} ~ \sup_{t \in [0,T]} \E^\delta[e^{-\gamma' Y_t^{0,Z}}] < \infty, \quad \text{for some}
\quad \gamma'>\gamma.
\end{equation}
Clearly, $\mathcal Z\neq\emptyset$ as it contains all bounded predictable processes and
 \begin{eqnarray*}
 \cal C
 &\supset&
 \Xi
 \;:=\;
 \big\{ Y_T^{Y_0,Z}:~Y_0\in \mathbb{R},~ Z \in \mathcal Z,~\mbox{and}~V_{\text{MM}}(Y_T^{Y_0,Z})\ge R \big\}. 
 \end{eqnarray*}
The next result shows that these sets are in fact equal, and identifies the market maker utility value and the corresponding optimal response. To prove equality of these sets, we are reduced to the problem of representing any contract $\xi\in\Cc$ as $\xi=Y^{Y_0,Z}_T$ for some $(Y_0,Z)\in\mathbb{R}\times\cal Z$, which is known in the literature as a problem of backward stochastic differential equation. We refrain from using this terminology, as our analysis does not require any result from this literature.
 
\begin{theorem}\label{thm:Agent} 
{\rm (i)} Any contract $\xi \in {\cal C}$ has a unique representation as $\xi=Y^{Y_0,Z}_T$, for some $(Y_0,Z)\in\mathbb{R}\times\cal Z$. In particular, ${\cal C}=\Xi$. \\
{\rm (ii)} Under this representation, the market maker utility value is 
 \begin{eqnarray*}
 V_{\rm MM}\big(\xi\big)=-e^{-\gamma Y_0},
 &\mbox{so that}&
 \Xi
 \;=\;
 \Big\{ Y_T^{Y_0,Z}:~Z \in \mathcal Z,~\mbox{and }~Y_0\ge \hat Y_0\Big\},
 ~~ \hat Y_0 := -\frac{1}{\gamma} \log(-R),
 \end{eqnarray*}
with the following optimal bid-ask policy
 \begin{equation} \label{optimalSpread}
 \hat\delta^i_t(\xi) = \Delta(Z^i_t), 
 ~i\in\{b,a\},
 ~\mbox{where}~
 \Delta(z) 
 :=  
 (-\delta_\infty) 
 \lor 
 \Big\{-z + \frac{1}{\gamma} \log\Big(1 + \frac{\sigma \gamma}{k}\Big)\Big\} 
 \land 
 \delta_\infty.
 \end{equation}
\end{theorem}

\noindent The proof of Part (i) is reported in Section \ref{sect:DPP}, and is obtained by using the dynamic continuation utility process of the market maker, following the approach of Sannikov \cite{sannikov2008continuous}. 

\vspace{5mm}

\noindent {\bf Proof of Theorem \ref{thm:Agent} (ii)} Let $\xi=Y_T^{Y_0,Z}$ with $(Y_0,Z)\in\mathbb{R}\times\cal Z$. We first prove that $J_{\text{MM}}(\delta, \xi)\le -e^{-\gamma Y_0}$ for all $\delta\in\cal A$. Denote $\overline{Y}_t:=Y^{Y_0,Z}_t+\sum_{i=a,b}\int_0^t \delta^i_tdN^i_t+Q_tdS_t$. Setting $h^\delta:=h(\delta,.)$, it follows from It\^o's formula that
 \begin{align*}
 d e^{-\gamma \overline{Y}_t}
 =
 \gamma e^{-\gamma \overline{Y}_{t-}} 
 &\Big[- (Q_t+Z^S_t) dS_t 
            - \!\!\sum_{i=b,a}\frac{1-e^{-\gamma (Z^i_t+\delta^i_t)}}{\gamma}d\widetilde N^{i,\delta}_t
            %\\ &
            + \big(H-h^{\delta_t}\big)(Z_t,Q_t)dt
                                            \Big],                                              
 \end{align*}
implying that $e^{-\gamma \overline Y}  $ is a $\P^\delta$-local submartingale. By Condition \eqref{conditionY}, the uniform boundedness of the intensities of $N^a$ and $N^b$ and H\"older inequality, $(e^{-\gamma \overline Y_t})_{t \in [0,T]} $ is uniformly integrable. By Doob-Meyer decomposition theorem, we conclude that $ 
 \int_0^\cdot \gamma e^{-\gamma \overline{Y}_{t-}} 
                        \big( - (Q_t+Z^S_t) dS_t 
                                - \sum_{i=b,a}\frac{1-e^{-\gamma (Z^i_t+\delta^i_t)}}{\gamma}d\widetilde N^{i,\delta}_t
                        \big), 
 $ is a martingale. It follows that
 $$
 J_{\text{MM}}(\delta, \xi)
 =
 \E^{\delta}\big[-e^{-\gamma \overline{Y}_T}\big] 
 =
 -e^{-\gamma Y_0}
 -\E^{\delta}\Big[\int_0^T \!\!\!\gamma e^{-\gamma \overline{Y}_t }
                                                                   \big(H(Z_t,Q_t)-h(\delta_t,Z_t,Q_t)\big) dt
                                               \Big]
 \le 
 -e^{-\gamma Y_0}.
 $$
On the other hand, equality holds in the last inequality if and only if $\delta$ is chosen as the maximizer of the Hamiltonian $H$ ($dt\times d\mathbb{P}^0-$a.e.), thus leading to the unique maximizer $\hat\delta(\xi)$ {given by \eqref{optimalSpread}}, which then induces $J_{\text{MM}}(\hat\delta(\xi), \xi)=-e^{-\gamma Y_0}$. This completes the proof that $V_{\text{MM}} (\xi)=-e^{-\gamma Y_0}$ with optimal response $\hat\delta(\xi)$.
\qed

%\begin{remark} \label{remark_importante}
%\noindent We have that for any $C > 0$
%\begin{equation} \label{finite}
%\sup_{t \in [0, T]} \sup_{ \delta \in \Ac} \E^\delta[\exp(C|Y_t|)] < \infty
%\end{equation}
%since $|V^\xi|$ admits finite moments. In fact, using the definition of $V^\xi$ in \eqref{V_t} together with Jensen's inequality we can show that for any $p > 1$
%$$ \E^\delta[|V^\xi_t|^p] \leq \sup_\delta\E^\delta\Big[ \exp\Big(-\gamma p \big(\int_t^T(\delta^a_u dN^a_u + \delta^b_u dN^b_u + q_u dS_u) + \xi\big)\Big) \Big] %. $$
%\end{remark}

\section{Designing the optimal contract}\label{opticont} 

\subsection{Risk-neutral exchange as a toy example.}

To understand the shape of the optimal contract, we first study the case where the exchange is risk neutral, corresponding to the limit where $\eta$ goes to $0$ for which we can derive the optimal compensation with explicit computations. In the present setting, we set $\overline q=+\infty$, thus relaxing the boundedness restriction on the inventory. By Theorem \ref{thm:Agent}, the problem of the exchange reduces to 
\begin{eqnarray}\label{PbRNex}
 V^E_0 
 = 
 \sup_{Y_0 \geq \hat Y_0} \sup_{Z \in {\cal Z}} 
 \E^{\hat\delta(Y^{Y_0,Z}_T)}\Big[c(N^a_T  + N^b_T)-Y^{Y_0,Z}_T\Big]
 = 
 \sup_{Z \in {\cal Z}} 
 \E^{\hat\delta(Y^{\hat Y_0,Z}_T)}\Big[c(N^a_T  + N^b_T)-Y^{\hat Y_0,Z}_T\Big],
 \end{eqnarray}
with $\hat \delta^i_t=  \Delta(Z^i_t)$, $t\in[0,T]$, $i\in\{a,b\}$, and where the maximization over $Y_0$ is achieved at $\hat Y_0$, due to the fact that  the market maker optimal response $\hat{\delta}(Y_T^{Y_0, Z})$, given by \eqref{optimalSpread}, does not depend on $Y_0$ so that the objective function is decreasing in $Y_0$.

 \begin{theorem} Consider the risk neutral exchange case $\eta\searrow 0$, and assume $\delta_\infty\ge  \frac{\sigma}k-\frac{\sigma}{k+\sigma\gamma} + \frac{1}{\gamma} \log(1 + \frac{\sigma \gamma}{k})-c$. Then the optimal contract for the exchange problem \eqref{PbRNex} is:
 \begin{eqnarray}\label{xistar:sol:rn}
 \hat \xi
 &=& 
 \hat Y_0+c(N_T^a+N_T^b)-\int_0^T  Q_r dS_r-\frac{\sigma T}{k+\sigma\gamma},
\end{eqnarray}
with optimal market maker effort: 
 \begin{equation*}\label{optimalOptimalSpread:rn}
 \hat\delta^a_t(\hat\xi) =  \hat\delta^b_t(\hat\xi)= \frac{\sigma}k-\frac{\sigma}{k+\sigma\gamma} + \frac{1}{\gamma} \log(1 + \frac{\sigma \gamma}{k})-c.
 \end{equation*}
 \end{theorem}
 
 \begin{proof}
By setting $\tilde c=c+\frac{\sigma}{k+\sigma\gamma}$, note that
 $$
 \E^{\hat\delta(Y^{\hat Y_0,Z}_T)}\Big[c\!\!\sum_{i=b,a}\!\!N^i_T-Y^{\hat Y_0,Z}_T\Big]\\
 = 
 \E^{\hat\delta(Y^{\hat Y_0,Z}_T)}\Big[\int_0^T \!\!\sum_{i=b,a}\!\!\lambda(\hat \delta^i_t)(\tilde c-Z_t^i)dt-\hat Y_0-\!\!\int_0^T\!\!  \Big(\frac12\gamma \sigma^2 (Z^S_r + Q_r)^2\Big)dr\Big],
 $$
so that the optimizer in \eqref{PbRNex} are given by $Z_r^{S,\star}=-Q_r,\; Z_r^{a,\star}= Z_r^{b,\star}=\tilde c- \frac{\sigma}k.$
 \end{proof}
 
\noindent Note that the optimal contract given by \eqref{xistar:sol:rn} emphasizes a risk transfer between the payoff of the market maker and that of the exchange through the term $\int_0^T  Q_r dS_r$.

\subsection{Exponential risk averse exchange}

By Theorem \ref{thm:Agent}, and solving the maximization with respect to $Y_0\ge\hat Y_0$ as in the previous subsection, the exchange problem \eqref{Ppb} reduces to
\begin{eqnarray}\label{Pb2}
 V^E_0 
 = 
 e^{\eta \hat Y_0}\; v^E_0,
 &\mbox{where}&
 v^E_0
 :=
 \sup_{Z \in {\cal Z}} 
 \E^{\hat\delta(Y^{\hat Y_0,Z}_T)}\Big[-e^{- \eta\big(c(N^a_T + N^b_T)-Y^{0,Z}_T\big)}\Big].
 \end{eqnarray}

\subsubsection{The HJB equation for the reduced exchange problem}

Our approach for the control problem $v^E_0$ of \eqref{Pb2} is to derive a solution $v$ of the corresponding HJB equation, and to proceed by the standard verification argument in stochastic control to prove that the proposed solution $v$ coincides with the value function $v^E_0$.

By the standard dynamic programming approach, the HJB equation for \eqref{Pb2} is
  \begin{equation} \label{HJB}
  \partial_t v(t,q) + H_E\big(q, v(t,q), v(t,q+1), v(t,q-1)\big)
  \;=\; 
  0, \quad q \in \{-\bar q, \cdots ,\bar q\}, \quad t\in [0,T),
  \end{equation}
with boundary condition $v\big|_{t=T}=-1$, with Hamiltonian $H_E: [-\bar q,\bar q]\times (-\infty,0]^3\rightarrow \mathbb R$:
 \begin{eqnarray} \label{HE}
 H_E(q, {y}, y_{+},y_{-}) 
 &=&
 H^1_E(q, y)  
 +\1_{\{q>-\bar q\}}  H^0_E(y,y_-)
 +\1_{\{q<\bar q\}} H^0_E(y,y_+),
 \end{eqnarray}
and
 \begin{eqnarray*}
 H^1_E(q, {y})
 \;=\; 
 \sup_{z_s \in \R} h^1_E(q,{y},z_s),
 &\!\!\mbox{and}&\!\!
 h^1_E(q,{y},z_s)
 =\frac{\eta \sigma^2}{2}  \;{y} \left( \gamma (z_s+ q)^2 + \eta z_s^2 \right),
 \\
 H^0_E({y},y')
 \;=\;
 \sup_{\zeta\in \R}
 h^0_E({y},y',\zeta)
 &\!\!\mbox{and}&\!\!
 h^0_E({y},y',{\zeta})
 =
 \lambda\big(\Delta(\zeta)\big) 
                                       \Big[y' e^{\eta(\zeta-c)} -{y} \big(1
                                                                                           +\eta \, \frac{1-e^{-\gamma(\zeta + \Delta(\zeta))}}
                                                                                                             {\gamma}
                                                                                    \big)
                                       \Big].
 \end{eqnarray*} 
By Lemma \ref{lemmaAnnex} below, the maximizers $\hat z = (\hat z^s, \hat z^a, \hat z^b)$ of $H_E$ are given by:
 \begin{eqnarray} 
 &\hat z^s(t, q)
 = 
 -\frac{\gamma}{\gamma + \eta} q,
 ~
 \hat z^a(t, q)
 = 
 \hat\zeta\big(v(t,q),v(t,q-1)\big),
 ~~
 \hat z^b(t, q)
 = 
 \hat\zeta\big(v(t,q),v(t,q+1)\big),
 ~~
 &
 \label{opt-z}
 \\
 &\mbox{with}~~
 \displaystyle\hat\zeta({y},y')
 =
 \zeta_0
 + \frac{1}{\eta} \log\Big( \frac{{y}}{y'}\Big),
 ~~
 \zeta_0
 \;=\;
 c + \frac{1}{\eta}\log\Big(1 - \frac{\sigma^2 \gamma \eta}
                                                               {(k + \sigma \gamma)(k + \sigma \eta)}
                                            \Big).
 &
 \nonumber
 \end{eqnarray}
Here, ${\delta_\infty}$ is sufficiently large so that Condition \eqref{condition:v1v2} of Lemma \ref{lemmaAnnex} is always met, namely 
\begin{equation} \label{cond_infty}
{\delta_\infty}  
\;\ge\; C_\infty + \frac{1}{\eta} \sup_{{t \in [0,T]}} \sup_{q \in [-\bar q, \bar q - 1]}\left| \log\left(\frac{v(t,q)}{v(t,q+1)}\right) \right|,
\end{equation}
with $C_\infty$ given in Lemma \ref{lemmaAnnex}, and we shall check in our verification argument that our candidate solution of the HJB equation will verify it. Using again the calculation reported in Lemma \ref{lemmaAnnex}, we rewrite the HJB equation \eqref{HJB} as
 \begin{equation} \label{edp_v}
 \partial_t v(t,q) 
 +  
 \frac{\gamma \eta^2 \sigma^2}{2(\gamma + \eta)} q^2 v(t,q) 
 - C_0 v(t,q) \Big[\1_{\{q>-\bar q\}} \big(\frac{v(t,q)}{v(t,q-1)}\big)^{\frac{k}{\sigma \eta}} 
                        + \1_{\{q<\bar q\}} \big(\frac{v(t,q)}{v(t,q+1)}\big)^{\frac{k}{\sigma \eta}} 
                \Big]
\;=\;
0, 
\end{equation}
with boundary condition $v\big|_{t=T}=-1$, where the constant $C_0$ is given by
 \begin{eqnarray*}
 C_0
 = C_0(\frac{\sigma\gamma}k,\frac{\sigma\eta}k),
 &\mbox{with}&
 C_0(\alpha,\beta)
 := A \beta (1+\alpha)^{-\frac1\alpha} \Big(1-\frac{\alpha\beta}{(1+\alpha)(1+\beta)}\Big)^{1+\frac1\beta}.
 \end{eqnarray*}
Inspired by \cite{gueant2013dealing}, we now make the key observation that this equation can be reduced to a linear equation by introducing $u := (-v)^{-\frac{k}{\sigma \eta}}$. By direct substitution, we obtain the following linear differential equation 
\begin{equation} \label{linearPDE}
u\big|_{t=T}=1,
~\mbox{and}~
\partial_t u(t,q) -F_{C_1,C_1'}(q,u(t,q),u(t,q+1),u(t,q-1))  
=
0,~ t<T,~|q|\le\bar q,
\end{equation}
$$ 
F_{m,m'}(q,y,y',y''):=m q^2 y- m'\big(y' \1_{\{q<\bar q\}} + y'' \1_{\{q>-\bar q\}}\big), 
~C_1 := \frac{k \gamma \eta \sigma}{2(\gamma + \eta)},
~C'_1 :=  \frac{kC_0}{\sigma \eta}.
$$
This equation can be written in terms of the $\R^{2\bar q+1}-$valued function $\mathbf{u}(t)=\big(u(t,q)\big)_{q \in \{-\bar q,\ldots,\bar q\}}$, of the variable $t$ only, as the linear ordinary differential equation
 \begin{eqnarray*}
 \partial_t\mathbf{u}
 =
 -\mathbf{B}u,
 &\mbox{where}&
 \mathbf{B}
 =
 {\tiny\left(
 \begin{array}{ccccc}
 -C_1{\bar q}^2    & C'_1           &                   &                   & 
 \\
 \ddots & \ddots & \ddots  &                   & 
 \\
                   & C'_1           & -C_1q^2    & C'_1            &  
 \\
                   &                   & \ddots & \ddots & \ddots  
 \\
                   &                   &                   &    C'_1       &    -C_1{\bar q}^2  
 \end{array}
 \right)}
 \leftarrow~q\mbox{-th line,} 
 \end{eqnarray*}
is a tri-diagonal matrix with lines labelled $-\bar q,\ldots,\bar q$. Denote by $\mathbf{b}_q$ the vector of $\R^{2\bar q+1}$ with zeros everywhere except at the position $q$, i.e. $\mathbf{b}_{q,j}=\1_{\{j=q\}}$ for $j \in \{-\bar q,\ldots,\bar q\}$, and $\mathbf{1}=\sum_{q=-\bar q}^{\bar q} \mathbf{b}_q$. Then, this ODE has a unique solution 
 \begin{equation}\label{u=etB}
 \mathbf{u}(t) = e^{(T-t)\mathbf{B}}\mathbf{1},
 ~\mbox{so that}~
 u(t,q)=\mathbf{b}_q\!\cdot\! e^{(T-t)\mathbf{B}}\mathbf{1},
 ~\mbox{and}~
 v(t,q)
 =
 -\big(\mathbf{b}_q\!\cdot\! e^{(T-t)\mathbf{B}}\mathbf{1}\big)^{-\frac{\sigma\eta}{k}}.
 \end{equation}
In the next section, we shall prove that this solution $v$ of the HJB equation \eqref{HJB} coincides with the value function of the reduced exchange problem \eqref{Pb2}, with optimal controls $\hat z(t,q)$ given in \eqref{opt-z}, thus inducing the optimal contract $Y^{\hat Y_0,\hat Z}_T$ with $\hat Z_t=\hat z(t,Q_{t-})$.

Let us notice that we may provide a more explicit expression of the above function $u$:
 \begin{equation}\label{calcul:u}
 \begin{array}{rcl}
 u(t,q)
 &=&
 \displaystyle
 \sum_{p\geq 0}  \frac{ [C'_1(T-t)]^p}{p!} \sum_{j\ge 0}  \frac{ [C'_1(T-t)]^j}{j!} e^{-C_1(T-t) (q+j-p)^2}
 \1_{\{|q+j-p|\le\bar q\}},
\end{array}
 \end{equation}
see Appendix \ref{Appendix-calcul u} for the more general case of $N$ symmetric exchanges in Nash equilibrium. We conclude this section by an (yet one more) alternative representation of the function $u$, which is convenient for the derivation of some useful properties.

\begin{prop}\label{prop:Representation}
Let $u$ and $v$ be defined by \eqref{u=etB}. The function $u$ can be represented as
 \begin{eqnarray*}
 u(t,q)
 &=&
 \E\Big[ e^{\int_t^T (-C_1(Q^{t,q}_s)^2+\overline{\lambda}_s+\underline{\lambda}_s)ds}\Big],
 \end{eqnarray*}
where $Q^{t,q}_s=q+\int_t^s d(\overline{N}_u-\underline{N}_u)$, and $(\overline{N},\underline{N})$ is a two-dimensional point process with intensity $(\overline{\lambda}_s, \underline{\lambda}_s) = C'_1 (\1_{\{Q_{s-} <\bar q\}}, \1_{\{Q_{s-} >-\bar q\}})$. In particular, we have $e^{-C_1{\bar q}^2 T}
 \le
 u
\le
 e^{2C'_1 T}$,
and Condition \eqref{cond_infty} is verified whenever
 \begin{equation} \label{bounds}
 \delta_\infty
 \;\geq\;
 \Delta_\infty := C_\infty + \frac{\sigma}{k} (2 C'_1 + C_1 \bar q^2 ) T.
\end{equation}
\end{prop}

\begin{proof} Notice that $u$ is a smooth bounded function. Denote $f(x) = -C_1 x^2 +C'_1 (\1_{\{x>-\bar q\}}+\1_{\{x<\bar q\}})$, and $ M_s = e^{\int_t^s f(Q^{t,q}_u) du} u(s, Q^{t,q}_s)$, $t\le s\le T$. We now show that $M$ is a martingale, so that $u(t,q)=M_t=\E[M_T]=\E\big[e^{-\int_t^T f(Q^{t,q}_s)ds}\big]$, as $u(T,.)=1$. To see that $M$ is a martingale, we compute by It\^o's formula that
 \begin{eqnarray*} 
 dM_s
 &=& 
 \big[u(s, Q^{t,q}_s)f(Q^{t,q}_s)  + \partial_tu(s, Q^{t,q}_s) \big]ds
 \\
 &&
 + C'_1 \big[u(s, Q^{t,q}_{s-} +1) - u(s, Q^{t,q}_{s-})\big] d\overline{N}_s 
 + C'_1 \big[u(s, Q^{t,q}_{s-} -1) - u(s, Q^{t,q}_{s-})\big]d\underline{N}_s  .
\end{eqnarray*}
Since $u$ is solution of \eqref{linearPDE}, we get
 \begin{eqnarray*} 
 dM_s 
 &=&  
 C'_1 \big[u(s, Q^{t,q}_{s-} +1) - u(s, Q^{t,q}_{s-})\big] d\overline{M}_s
 + C'_1 \big[u(s, Q^{t,q}_{s-} -1) - u(s, Q^{t,q}_{s-})\big]d\underline{M}_s,
 \end{eqnarray*}
where $(\overline{M}, \underline{M}) = (\overline{N} - \int_0^\cdot \overline{\lambda}_s ds, \underline{N} - \int_0^\cdot \underline{\lambda}_s ds)$ is a martingale. The martingale property of $M$ now follows from the boundedness of $u$ as it can be verified from the expression \eqref{u=etB}.  Finally, the bound $|Q^{t,q}_s|\le \bar q$ induces directly the announced bounds on $u$, which in turn imply Condition \eqref{cond_infty} when \eqref{bounds} is satisfied because $v = - u^{-\frac{\sigma \eta}{k}}$ .\end{proof}

\subsubsection{Main result}

We now verify that the function $v$ derived in the previous section is the value function of the exchange, with optimal feedback controls $(\hat z^s,\hat z^a,\hat z^b)$ as given in \eqref{opt-z}, thus identifying a unique optimal contract to be proposed by the exchange to the market maker. 

\begin{theorem}\label{thm:pbPlatform} Assume that $\delta_\infty\ge\Delta_\infty$, with $\Delta_\infty$ given by \eqref{bounds} and define $u$ and $v$ by \eqref{u=etB}. Then the optimal contract for the problem of the exchange \eqref{Ppb} is given by
 \begin{eqnarray}\label{xistar:sol}
 \hat \xi
 &=& 
 \hat Y_0 
 + \int_0^T \hat Z_r^a dN_r^{a} + \hat Z^b_r dN_r^{b} + \hat Z_r^S dS_r 
                  +\Big( \frac12\gamma\sigma^2 \big(\hat Z_r^{S}+Q_r\big)^2 - H\big(\hat Z_r, Q_r\big)\Big) dr,
\end{eqnarray}
with $\hat Z^S_r=\hat z^s(r, Q_{r-})$, $\hat Z^a_r=\hat z^a(r,Q_{r-})$, and $\hat Z^b_r=\hat z^b(r,Q_{r-})$ as defined in \eqref{opt-z}. The market maker's optimal effort is given by 
 \begin{equation}\label{optimalOptimalSpread}
 \hat\delta^a_t = \hat\delta^a_t(\hat\xi) = -\hat Z_t^a + \frac{1}{\gamma} \log(1 + \frac{\sigma \gamma}{k}), 
 \quad 
 \hat\delta^b_t = \hat\delta^b_t(\hat\xi)=  -\hat Z_t^b + \frac{1}{\gamma} \log(1 + \frac{\sigma \gamma}{k}).
 \end{equation}
\end{theorem}

\begin{remark}
Notice that, in our model the exchange observes the spread set by the market maker. However, as explained above, the spread cannot be part of the contract. Consequently, the second best exchange problem in Theorem \ref{thm:pbPlatform} does not coincide with the first best where the exchange could use the observe bid-ask policy $\delta$ in the contract $\xi$, under the market maker participation constraint. The corresponding computations are reported in Appendix \ref{sec:firstbest} below.
\end{remark}

\subsubsection{Discussions and interpretations}\label{discussion}

The processes $\hat Z^a$, $\hat Z^b$ and $\hat Z^S$ defining the optimal contract have natural interpretations. Based on Proposition \ref{prop:Representation}, we can get the intuition that (at least for large inventories) \begin{equation}\label{eq:zstar:eqQ}
 \hat Z^i= \xi_0+\frac1\eta\log\Big( \frac{u(t,Q_{t-})}{u(t,Q_{t-}-\eps_i)}\Big) \underset{|q|\to+\infty}{\sim}\xi_0 +\frac{\eps_i}{\eta} \frac{C_1}kQ_{t-},\; i\in \{a,b\},
\end{equation}
recalling that $(\varepsilon_b,\varepsilon_a)=(-1,1)$.
 This is confirmed in our Figure \ref{spreadaskbid} below at time $t=0$ (since $\hat Z^b$ and $\hat Z^a$ are the opposite of the optimal bid and ask spreads respectively). This is in fact shown for any time in the numerical simulations and asymptotic expansion in \cite[Section 4]{gueant2013dealing} and \cite[Section 3.2]{avellaneda2008high} where same type of PDE as ours is considered. Thus, when the inventory is highly positive, the exchange provides incentives to the market-maker
so that it attracts buy market orders and discourage him from more sell market orders, and vice versa for a negative inventory. The integral 
$\int_0^T \hat Z_r^S dS_r$ can be understood as a risk sharing term. More precisely, $\int_0^t Q_{r} dS_r$ corresponds to the price driven component of the
inventory risk $Q_tS_t$. Hence, the {exchange} supports the proportion $\frac{\gamma}{\gamma + \eta}$ of this risk so that the market maker maintains reasonable quotes despite some inventory. 

Notice that for a highly risk averse exchange, i.e. $\eta\nearrow\infty$, 
\[  \int_0^T \hat Z_r^a dN_r^{a} + \hat Z^b_r dN_r^{b}\approx c(N_T^a+N_T^b), \; \hat Z_r^S\approx 0,\]
meaning that the exchange transfers to the market maker the total fee. This is the so-called \textit{selling the firm} effect, as the exchange delegates all benefit to the market maker.\footnote{We would like to thank an anonymous referee for suggesting this interpretation.}

Until now, we have focused on the maker part of the make-take fees problem since we have considered that the taker cost $c$ is fixed. Nevertheless, our approach also enables us to suggest the exchange a relevant value for $c$. Actually, we see that when acting optimally, the exchange transfers the totality of the fixed taker fee $c$ to the market maker. It is therefore neutral to the value of $c$ as its optimal utility function $v_0^E = v(0, Q_0)$ is independent of the taker cost, see \eqref{edp_v}.  However, $c$ plays an important role in the optimal spread offered by the market maker given by
 $$ 
 - 2 c + \frac{\sigma}{k} \log\Big( \frac{u(t,Q_{t-})^2}{u(t,Q_{t-}-1) u(t,Q_{t-}+1)}\Big) 
 - \frac{2}{\eta}\log\Big(1 - \frac{\sigma^2 \gamma \eta}{(k + \sigma \gamma)(k + \sigma \eta)}\Big) 
 + \frac{2}{\gamma} \log(1 + \frac{\sigma \gamma}{k}). 
 $$
Furthermore, from numerical computations\footnote{See indeed Figure \ref{spread_t} by noting that $u$ does not depend on the fee $c$.} or the asymptotic development \eqref{eq:zstar:eqQ}, we remark that
$  \frac{u(t,q)^2}{u(t, q-1) u(t,q+1)}$
is close to unity for any $t$ and $q$. Hence if for example the exchange targets a spread close to one tick (see \cite{dayri2015large,huang2016predict} for details on optimal tick sizes and spreads), it can be obtained by setting
$$ c \approx - \frac{1}{2} \text{Tick}-  \frac{1}{\eta}\log\Big(1 - \frac{\sigma^2 \gamma \eta}{(k + \sigma \gamma)(k + \sigma \eta)}\Big) + \frac{1}{\gamma} \log(1 + \frac{\sigma \gamma}{k}).$$
For $\sigma \gamma/k$ small enough, this equation reduces to
\begin{equation}\label{optimalc}
c \approx \frac{\sigma}{k}- \frac{1}{2} \text{Tick}.
\end{equation}
This is a particularly simple formula for setting the taker constant fee $c$, as the parameters $\sigma$ and $k$ can be easily estimated from market data. We see that the higher the volatility, the larger the taker cost should be. The decrease in $k$ is also natural: If $k$ is large, the liquidity vanishes rapidly when the spread becomes wide, meaning that market takers are sensitive to extra costs relative to the efficient price. Therefore, the taker cost has to be small if the exchange wants to maintain a reasonable market order flow. 

\section{Exchange impact on market quality}
\label{section:comparison}

In this section, we compare our setting with the situation without incentive policy from an exchange towards market making activities which corresponds to the problem of optimal market making considered in \cite{avellaneda2008high, gueant2013dealing}. The results in \cite{avellaneda2008high} are taken as benchmark for our investigation to emphasize the impact of the incentive policy on market quality. We will refer to this case as the neutral exchange case.

Let us first recall the results in \cite{avellaneda2008high,gueant2013dealing}. The optimal controls of the market maker denoted by $\widetilde{\delta}^{a}$ and $\widetilde{\delta}^{b}$ are given as a function of the inventory $Q_t$ by
 $$ 
 \widetilde{\delta}^{i}_t
 = 
 \frac{\sigma}{k} \log\Big( \frac{\widetilde{u}(t,Q_{t-})}{\widetilde{u}(t,Q_{t-}-\varepsilon_i)}\Big) 
 + \frac{1}{\gamma}\log(1 +\frac{\sigma \gamma}{k} ), 
 ~~i\in\{b,a\},~
 (\varepsilon_b,\varepsilon_a)=(-1,1), 
 $$
where $\widetilde{u}$ is the unique solution of the linear differential equation
$$
\widetilde{u}\big|_{t=T}=1
~\mbox{and}~
\partial_t \widetilde{u}(t,q) - F_{\widetilde{C_1}, \widetilde{C'_1}}(q,\widetilde{u}(t,q),\widetilde{u}(t,q+1),\widetilde{u}(t,q-1))  = 0 , ~~t<T,|q|\le\bar q,
$$
with $\widetilde{C}_1 = \frac{\sigma \gamma k }2$ and $\widetilde{C}'_1 = A (1+\frac{\sigma \gamma}{k})^{-(1+\frac{\sigma \gamma}k)}$. In our case, the optimal quotes $\hat\delta^a$ and $\hat\delta^b$ are obtained from Theorem \ref{thm:pbPlatform} and satisfy for $i\in\{b,a\}$, and
 $(\varepsilon_b,\varepsilon_a)=(-1,1)$:
 $$ 
 \hat\delta^i_t
 =
 \frac{\sigma}{k} \log\Big( \frac{{u}(t,Q_{t-})}{{u}(t,Q_{t-}-\varepsilon_i)}\Big) 
 + \frac{1}{\gamma}\log(1 +\frac{\sigma \gamma}{k} ) 
 -c - \frac{1}{\eta} \log\Big(1- \frac{\sigma^2 \gamma \eta}{(k+ \sigma \gamma)(k+ \sigma \eta)}\Big). 
 $$
where $u$ is solution of the linear equation \eqref{linearPDE}.

Numerical experiments show that $u$ and {$\widetilde u$} decrease quickly to zero when $q$ becomes large, inducing numerical instabilities in the computation of 
$$ v_+(t,q)= \log\Big(\frac{u(t,q+1)}{u(t,q)} \Big), \quad  \widetilde{v}_+(t,q)= \log\Big(\frac{\widetilde{u}(t,q+1)}{\widetilde{u}(t,q)} \Big), \quad  q \in  \{ -\bar q, \cdots, \bar q-1\}, $$
which are crucial in the expressions of optimal quotes. To circumvent this numerical difficulty, we remark that $v_+$ and $\widetilde{v}_+$ are solution of the following integro-differential equations
\begin{eqnarray}
v_+\big|_{t=T}=0
&\mbox{and}&
\partial_t v_+(t,q) +  \mathcal F_{C_1,C'_1} (q,v_+(t,q),v_+(t,q+1), v_-(t,q+1)  ) = 0,
 \label{linearPDE1}
 \\
\widetilde{v}_+\big|_{t=T}=0
&\mbox{and}&
\partial_t \widetilde{v}_+(t,q)  + \mathcal F_{\widetilde{C_1},\widetilde{C_1'}} (q,\widetilde{v}_+(t,q),\widetilde{v}_+(t,q+1),\widetilde{v}_+(t,q-1) )= 0,
 \label{linearPDE2}
 \end{eqnarray}
where, again with $(\varepsilon_b,\varepsilon_a)=(-1,1)$,
\[ \mathcal F_{\alpha,\beta}(q,y,y_+,y_-)= {\alpha} (2q+1) - \beta \sum_{i\in\{a,b\}}\varepsilon_i e^{\varepsilon_i v_+(t,q-\varepsilon_i)} \1_{\{\varepsilon_i q<\bar q-1\}}-\varepsilon_i e^{\varepsilon_iv_+(t,q)}.\]  
We thus rather apply classical finite difference schemes to \eqref{linearPDE1} and \eqref{linearPDE2}.

In the following numerical illustrations, in the spirit of \cite[Section 6]{gueant2013dealing}, we take $T = 600s$ for an asset with  volatility $\sigma = 0.3 \text{ Tick}.s^{-1/2}$ (unless specified differently). Market orders arrive according to the intensities \eqref{intensity} with $A = 1.5 s^{-1}$ and $k=0.3 s^{-1/2}$. We assume that the threshold inventory of the market maker is $\bar q = 50$ units and we set his risk aversion parameter to $\gamma = 0.01$. The exchange is taken more risk averse with $\eta =1$. Finally, we assume that the taker cost $c = 0.5 \text{ Tick}$\footnote{Remark that the taker cost is chosen according to Criteria \eqref{optimalc}. We expect the optimal spread to be close to one tick. Note also that here the tick is just a unit and not a true market parameter.}. 

\subsection{Impact of the exchange on the spread and market liquidity}\label{spreadSec}

We start by comparing the optimal spread $ \hat\delta^{a}_0 + \hat\delta^{b }_0$ at time $0$ obtained when contracting optimally with the optimal spread $ \widetilde{\delta}^{a}_0 + \widetilde{\delta}^{b }_0$ without contracting. The optimal spreads are plotted in Figure \ref{spread1} for different initial inventory values $Q_0 \in \{-\bar q, \cdots, \bar q \}$.

\begin{center}
\includegraphics[scale=0.45]{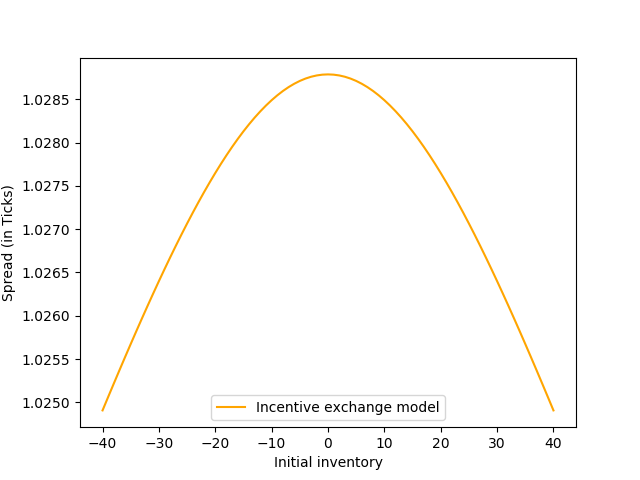}
\includegraphics[scale=0.45]{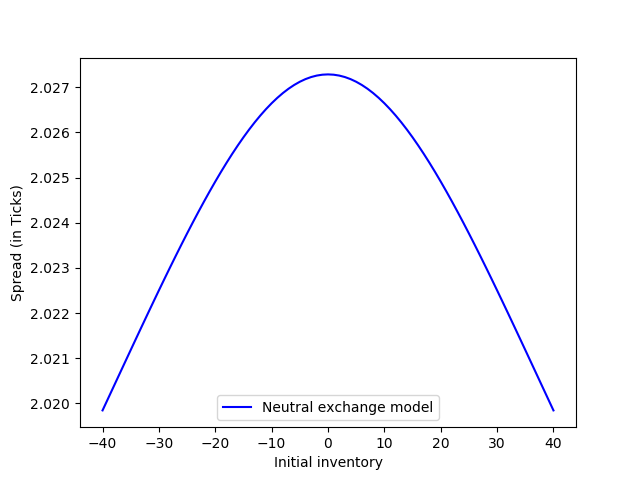}
\captionof{figure}{\small \it Comparison of optimal initial spreads with/without incentive policy from the exchange.}
\label{spread1}
\end{center}

\noindent  We observe in Figure \ref{spread1} that the initial spread does not depend a lot on the initial inventory {(because the considered time interval $[0,T]$ is not too small)} and that it is reduced thanks to the optimal contract between the market maker and the exchange. This is not surprising since in our case the exchange aims at increasing the market order flow by proposing an incentive contract to the market maker inducing a spread reduction. Actually this phenomenon occurs over the whole trading period $[0,T]$. To see this, we generate $5000$ paths of market scenarios and compute the average spread over $[0,T]$ for an initial inventory $Q_0=0$. The results are given in Figure \ref{spread_t}.
\begin{center}
\includegraphics[scale=0.45]{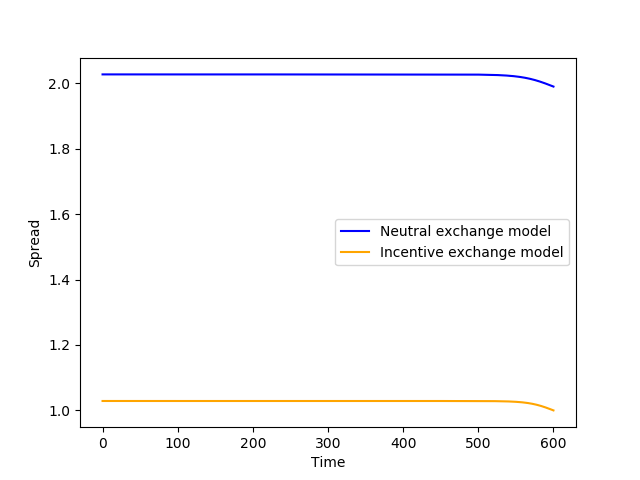}
\captionof{figure}{\small\it Average spread on $[0,T]$ with $95\%$ confidence interval, with/without incentive.} 
\label{spread_t}
\end{center}
Since the spread is tighter during the trading period under an incentive policy from the exchange, the arrival intensity of market orders is more important and hence the market is more liquid as shown in Figure \ref{order_flow}.
\begin{center}
\includegraphics[scale=0.45]{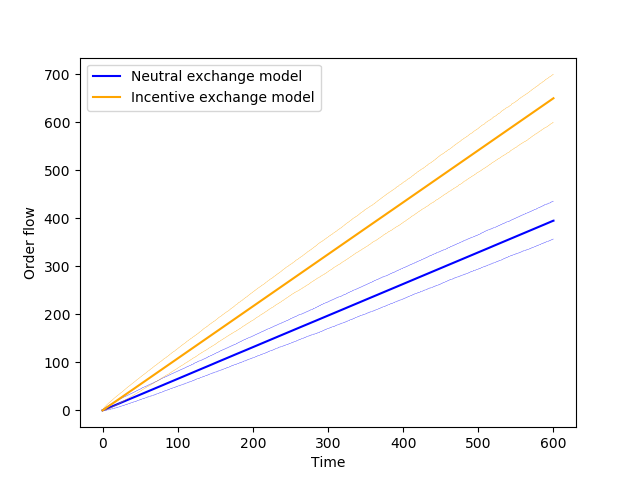}
\captionof{figure}{\small \it Average order flow on $[0,T]$ with $95\%$ confidence interval, with/without incentive.} 
\label{order_flow}
\end{center}

\noindent We now consider in Figure \ref{spreadaskbid} the bid and ask sides separately. We see that when the inventory is positive and very large, $\hat\delta^{a}$ and $\widetilde{\delta}^{a }$ are negative, meaning that the market maker is ready to sell at prices lower than the efficient price in order to attract market orders and reduce his inventory risk. On the contrary, if the inventory is negative and very large, in both situations, its ask quotes are well above the efficient price in order to repulse the arrival of buy market orders. However, since in our case the exchange remunerates the market maker for each arrival of market order, we get that the ask spread with contract $\hat\delta^{a }$ is smaller than $\widetilde{\delta}^{a }$. A symmetric conclusion holds for the bid part of the spread.\\

\begin{center}
\includegraphics[scale=0.45]{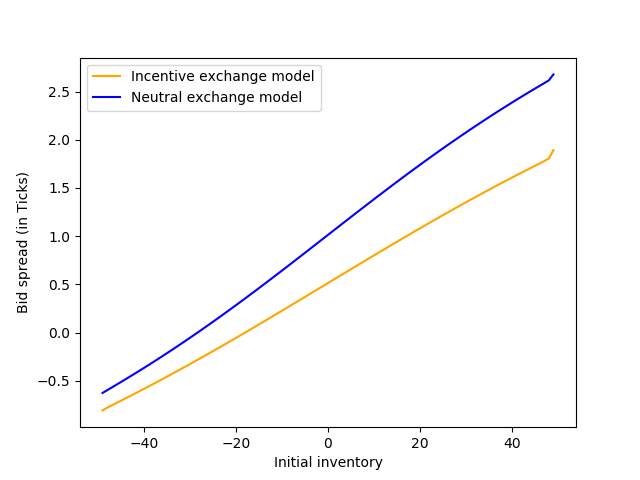}
\includegraphics[scale=0.45]{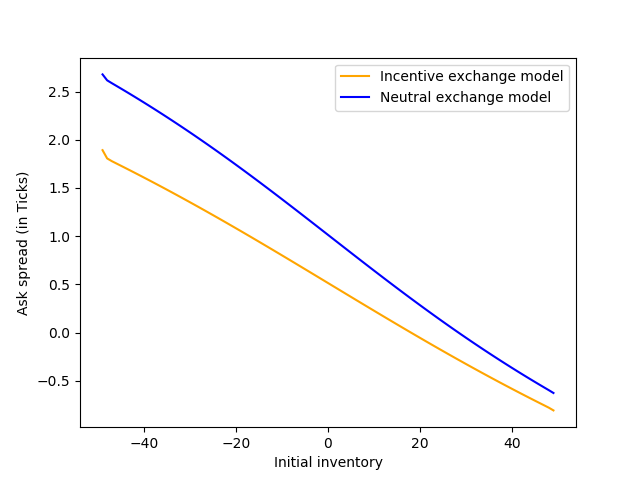}
\captionof{figure}{\small \it Optimal ask and bid spreads, with/without incentive policy.}
\label{spreadaskbid}
\end{center}

\noindent We now turn to the impact of the volatility on the spread. The optimal contract obtained in \eqref{xistar:sol} induces an inventory risk sharing phenomenon through the term $\hat Z^{S }$. Hence, when the volatility increases, the spread difference between situations with/without incentive policy becomes less important, see Figure \ref{spread_sigma} in which we consider the optimal initial spread difference when the initial inventory is set to zero between both situations with/without incentive policy from the exchange to the market maker for different values of the volatility.
\begin{center}
\includegraphics[scale=0.45]{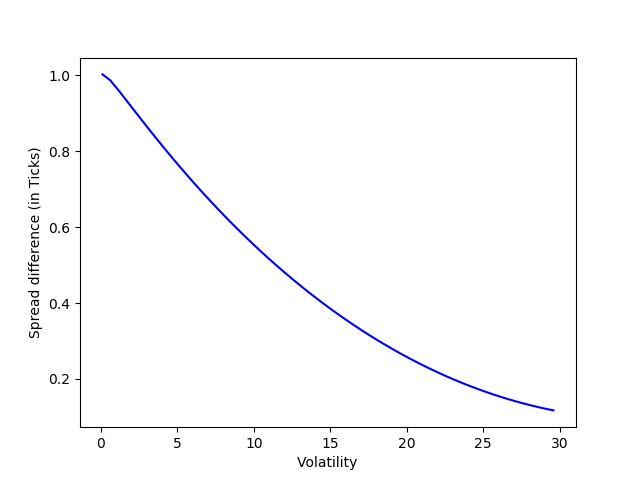}
\captionof{figure}{\small \it Initial optimal spread difference between the situations with and without incentive.} 
\label{spread_sigma}
\end{center}

\subsection{Impact on the P\&L of the exchange and the market maker}

We assume that $Q_0 = 0$. Recall that $ \mbox{\rm PL}^{\delta}$ defined in \eqref{PL} denotes the trading part of the profit and loss (P\&L) of the market maker for a given strategy $\delta$. In our case, the underlying total P\&L at time $t$ of a market maker acting optimally, denoted by $\mbox{\rm PL}_t^\star$, is:
 $$ 
 \mbox{\rm PL}_t^\star= \mbox{\rm PL}^{\hat\delta}_t + Y^{\hat Y_0,\hat Z}_t,
 $$ 
where $Y^{\hat Y_0, \hat Z}_t$ corresponds to the quantity on the right hand side of \eqref{xistar:sol} with $T$ replaced by $t$. We now compare this quantity to the benchmark $\mbox{\rm PL}^{\widetilde{\delta}}_t$ which corresponds to the optimal profit and loss without intervention of the exchange.

To make $ \mbox{\rm PL}_t^\star$ and $\mbox{\rm PL}^{\widetilde{\delta}}_t$ comparable, we choose $\hat Y_0$ in $\eqref{xistar:sol}$ so that the market maker gets the same utility in both situations, that is $ \hat Y_0 = \frac{k}{\sigma} \log(\widetilde{u}(0, Q_0))$. Thus, the market maker is indifferent between the situation with or without exchange intervention. We generate $5000$ paths of market scenarios and compare the average of both P\&L in Figure \ref{PL_market_maker} with and without incentive policy.

\begin{center}
\includegraphics[scale=0.45]{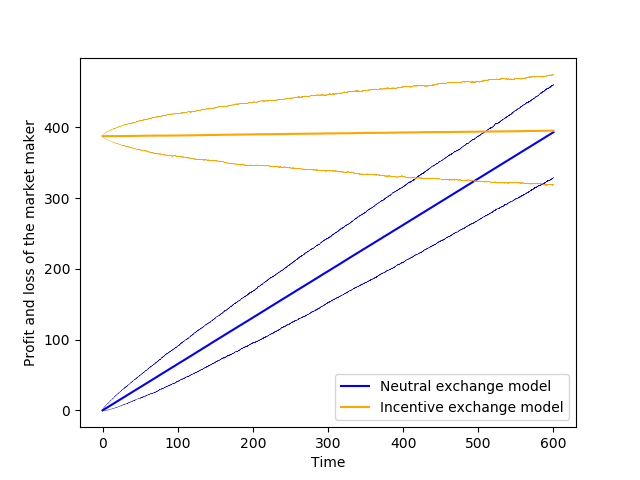}
\captionof{figure}{\small \it Average P\&L of market maker with/without incentive, with $95\%$ confidence interval.} 
\label{PL_market_maker}
\end{center}
Since $\hat Y_0$ is set to obtain the same utility in both cases, the two average P\&L are very close at the end of the trading period. The variance of the P\&L also seems to be the same in both situations. The only difference from the market maker viewpoint here is that in the case of a contract, the P\&L is already made at time $0$ thanks to the compensation of the exchange and then fluctuates slightly. This is because he is earning the spread but paying continuous ``coupons" $ \big(H(\hat Z_t, Q_t)  - \frac{\sigma^2 \gamma}{2} (\hat Z^S_t + Q_t)^2 \big)dt$ from the contract. In the case without exchange intervention, the market maker increases his P\&L over the whole trading period thanks to the spread.

We now compare the profit and loss of the exchange in the two considered cases. 
When it applies an incentive policy towards the market maker, the P\&L of the exchange is given by $c(N^a_t + N^b_t) - Y^{\hat Y_0, \hat Z}_t.$
When the exchange is neutral, its P\&L is simply $c(N^a_t+ N^b_t)$. We compare these two quantities in Figure \ref{PL_platform}.

\begin{center}
\includegraphics[scale=0.45]{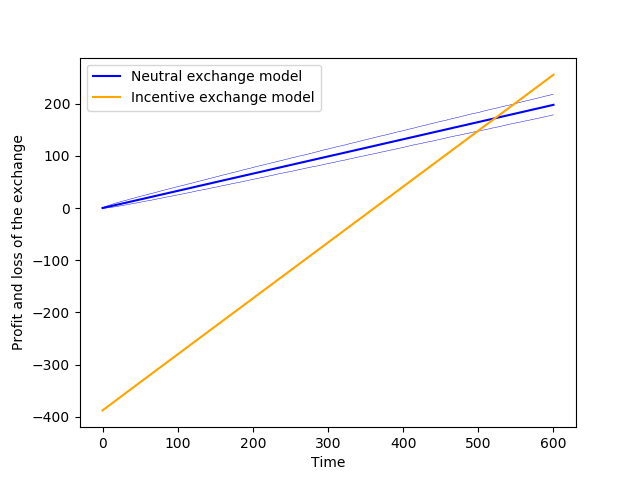}
\captionof{figure}{\small \it  Average P\&L of the exchange with/without incentive, with $95\%$ confidence interval.} 
\label{PL_platform}
\end{center}
We see that the initial P\&L of the contracting exchange is negative because of the initial payment $\hat Y_0$. However it finally exceeds, with a smaller standard deviation, the P\&L in the situation without incentive policy from the exchange. Hence the incentive policy of the exchange proves to be successful. Both configurations are indeed equivalent for market makers but the exchange obtains more revenues when contracting optimally. This is due to the fact that the contract triggers more market orders.

Finally, we plot the aggregated average P\&L of the market maker and the exchange (independent of the choice of the initial payment). We observe that it is always greater in the optimal contract case. 
\begin{center}
\includegraphics[scale=0.45]{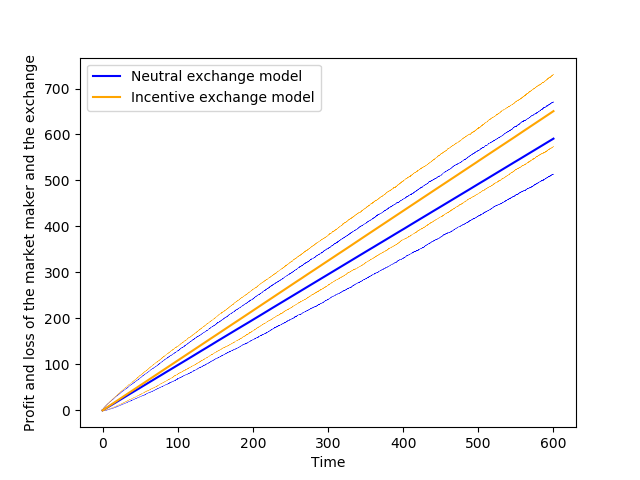}
\captionof{figure}{\small \it P\&L of exchange and market maker with/without incentive, $95\%$ confidence interval.} 
\label{PL_total}
\end{center}

\subsection{Impact of the incentive policy on the trading cost} 
We consider one single market taker. In the case without exchange, with the specified parameters and under optimal reaction of the market maker, this investor buys on average $200$ shares over $[0,T]$. To make the comparison with the case with exchange intervention, we modify the parameter $A$ appearing in the intensity \eqref{intensity} when simulating a market with optimal contract. This new value is chosen so that the investor buys on average the same number of assets $(200)$ over the time period. This amounts to take 
$A =  0.9s^{-1}$. We confirm in Figure \ref{ask_order_flow} that the average ask order flows agree in both situations.
\begin{center}
\includegraphics[scale=0.45]{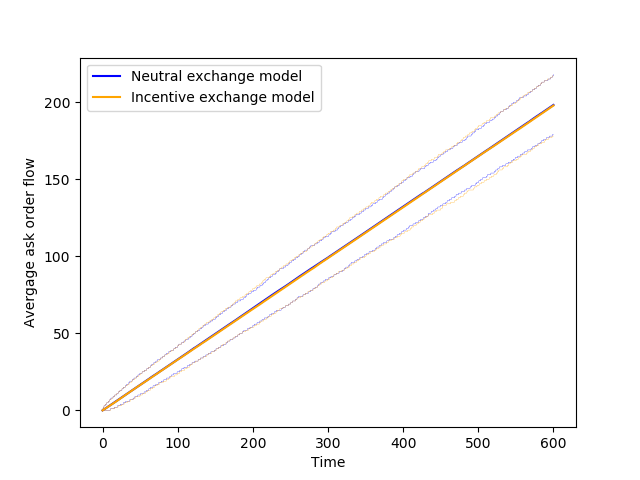}
\captionof{figure}{\small \it Setting similar average ask order flows on $[0,T]$ by taking different intensity basis $A$ in the case with and in the case without incentive policy; $95\%$ confidence interval.} 
\label{ask_order_flow}
\end{center}

Finally, Figure \ref{trading_cost} compares the average cost of trading for the market taker
$\E^\delta \big[\int_0^T \delta^a_t dN^a_t\big],$
with and without incentive, and shows that the reduced spreads lead to significantly smaller trading costs for investors.
\begin{center}
\includegraphics[scale=0.45]{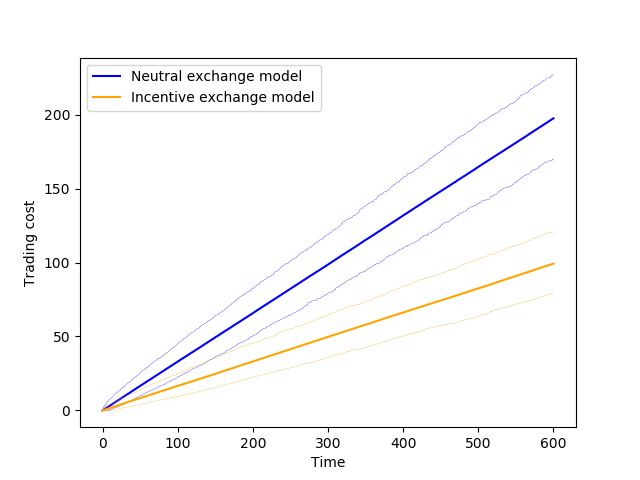}
\captionof{figure}{\small \it Average trading cost on $[0,T]$ with $95\%$ confidence interval, with/without incentive.} 
\label{trading_cost}
\end{center}

\section{Extension: symmetric exchanges competition}\label{section:oligo}

In this section we extend the previous study by considering a first step towards the investigation of the case of several exchanges in competition. 

\subsection{Symmetric exchanges in Nash equilibrium}

We assume here that $N$ identical exchanges display the quotes of one market maker and that the trading flows are split equally between the exchanges. More precisely, each time the market taker acts on the market, his trade of size one is split into $N$ trades of size $1/N$ distributed across all exchanges. This is equivalent to modify the market taker fee received by each exchange from $c$ to $c/N$. This situation is of course quite stylized but understanding it is obviously a very important preliminary to the study of the case of different exchanges with various market makers. Furthermore, as we will see below, the situation considered here is already significantly more intricate than the case of one exchange treated in the previous sections.

The market maker receives the aggregation of the compensation given by the $N$ exchanges denoted by $\overline \xi=\xi+\tilde \xi$, where $\xi$ and $\tilde \xi$ are repsectively the remuneration given by a representative exchange and the aggregation of the $N-1$ others. Hence, $\overline \xi$ inherits all the technical assumptions made previously on $\xi$ (for only one exchange), since the problem of the market maker is similar by considering $\overline \xi$ for his compensation. Consequently, the market maker's problem returns an optimal spread $\hat \delta(\overline \xi)$ so that Theorem \ref{thm:Agent} holds by considering $\overline \xi$. In view of the symmetry assumption made on the exchanges, any exchange aims at solving
 \begin{equation} \label{Ppbi} 
 V^{E}_0(\tilde \xi) =
 \sup_{\xi \in {\cal C}} ~ \E^{\hat\delta(\xi+\tilde\xi)}
                                     \Big[-e^{-\eta(\frac{c}N(N^a_T +N^b_T) - \xi)}\Big],
\end{equation}
where $\tilde \xi$ is fixed, and $\eta>0$ is the common risk aversion parameter of the $N$ exchanges.

\begin{definition}[Nash equilibrium and symmetric Nash equilibrium]\label{def:nash:sym}
A $N$-tuple $( \xi^{\rm e})_{1\leq {\rm e}\leq N}$ is a Nash equilibrium if for any ${\rm e}\in\{1,\dots,N\}$ we have
 \[V^{E}_0( \xi^{\rm e})
 = 
 \E^{\hat\delta( \sum_{j=1}^N \xi^j})
 \Big[-e^{-\eta(\frac{c}N(N^a_T +N^b_T) - \xi^{\rm e})}\Big].
 \]
 A $N$-tuple of contracts $(\xi^{\rm e})_{1\leq {\rm e}\leq N}$ is a symmetric Nash equilibrium if $( \xi^{\rm e})_{1\leq {\rm e}\leq N}$ is a Nash equilibrium such that $\xi^1=\cdots=\xi^N$. We denote by $\mathcal S^N:=\big\{\xi^0: (\xi^0,\ldots,\xi^0)\in\mathcal C^N\big\}$ the collection of all such symmetric Nash equilibria.
\end{definition}

From Theorem \ref{thm:Agent}, it follows that any symmetric Nash equilibrium  $\xi^0\in \mathcal S^N$ is induced by a pair $(\tilde y_0,\tilde Z)\in [\hat Y_0,+\infty)\times \mathcal Z$ such that 
\begin{eqnarray} 
\xi^0
\;=\; 
\frac{1}{N} Y_T^{\tilde y_0,\tilde Z}
&=&
\frac{\tilde y_0}N
+\int_0^T \frac{1}{N}\tilde Z_r d\chi_r
+\frac{\gamma \sigma^2}{2N} (\tilde Z_r^S+Q_r)^2 dr
-\frac{1}{N}\sum_{i\=a,b} H^i(\tilde Z_r^i,Q_r) dr,
\nonumber \\
&=&
\frac{\tilde y_0}N
+\!\!\int_0^T \!\!\zeta^0_r d\chi_r
+ \frac{\gamma \sigma^2}{2N} (N \zeta_r^{S,0}+Q_r)^2dr
-\frac 1N\!\!\sum_{i= a,b}\!\! H^i(N \zeta_r^{i,0},Q_r) dr,
\label{eq:nashsym}
\end{eqnarray}
with $\zeta^0=\frac{\tilde Z}N$, and
\begin{eqnarray*}
H^{i}(z,q)
= 
\lambda(\hat\delta(z)) \frac{\sigma}{k+\sigma\gamma}\mathbf 1_{\varepsilon_i q<\overline Q},
&\mbox{with}& 
(\varepsilon_b,\varepsilon_a)=(-1,1).
\end{eqnarray*}
We now denote by $\xi^{0,N-1}$ the $(N-1)$-tuple of identical contracts $\xi^0$ defined by \eqref{eq:nashsym}, and we set $\tilde Y_0:=\frac{N-1}N\tilde y_0$. As $\xi+(N-1)\xi^0=Y_T^{Y_0+\tilde Y_0,Z}$, by setting $\zeta:=Z-(N-1)\zeta^0$, for some $( Y_0, Z)\in [\hat Y_0,+\infty)\times \mathcal Z$  the problem of each exchange reduces to 
 \begin{equation*} \label{Ppbi:ter} 
 V^E_0(\xi^{0,N-1}) =
 \sup_{Y_0,\zeta} ~ \E^{\hat\delta(\zeta+(N-1) \zeta^0)} 
                                     \Big[-e^{-\eta( \int_0^T(\underset{i=a,b}{\sum}(\frac cN- \zeta_t^{i} )dN^i_t-\zeta_t^SdS_t)-Y_0-\int_0^T G(\zeta_t, \zeta^0_t, Q_t) - \alpha^0_t dt    )}\Big],
\end{equation*}
where $Y_0$ ranges in $[\hat Y_0-\tilde Y_0,+\infty)$, $\zeta\in\mathcal Z$, and 
 $$
 \begin{array}{rll}
 G(\zeta, \zeta^0, q)&=&  \frac12 \gamma \sigma^2(\zeta^{S}+(N-1) \zeta^{S,0}+q)^2-\underset{{i=a,b}}{\sum} H^i(\zeta^{i}+(N-1) \zeta^{i,0},q), 
 \\
\alpha^0_t& =& \frac{N-1}N \Big(\frac12 \gamma \sigma^2(N\zeta^{S,0}_t+Q_t)^2-\underset{{i=a,b}}{\sum}H^i(N{\zeta}^{i,0}_t,Q_t)\Big).
\end{array}
$$
The optimization over $Y_0$ is immediately solved, leading to
 \begin{equation} \label{Ppbi:ter} 
 V^E_0(\xi^{0,N-1}) =
 \sup_{\zeta \in  \mathcal Z} ~ \E^{\hat\delta(\zeta+(N-1) \zeta^0)} 
                                     \Big[-e^{-\eta( \int_0^T(\underset{i=a,b}{\sum}(\frac cN- \zeta_t^{i} )dN^i_t-\zeta_t^SdS_t)-Y^\star_0 -\int_0^T G(\zeta_t,\zeta^0_t, Q_t)   - \alpha^0_t dt    )}\Big],
\end{equation}
with $Y^\star_0=\hat Y_0-\tilde Y_0$. 

\begin{definition}[Markovian Nash equilibrium]\label{def:markovnash}
A symmetric Nash equilibrium $\xi^0\in \mathcal S^N$ is Markovian if the coefficients $\zeta^0$ appearing in $\eqref{eq:nashsym}$ is given by $\zeta^0_t=\zeta^0(t,Q_t)$ for some deterministic function $\zeta^0$.   
\end{definition}

\begin{remark}\label{remark:symmetric}
Note that if $ \xi^0$ is a symmetric Nash equilibrium with decomposition \eqref{eq:nashsym}, we necessarily have $\tilde y_0=\hat Y
_0$, $\hat{\zeta}(  \zeta^0)= \zeta^0$, where $\hat{\zeta}( \zeta^0)$ denotes an optimizer of \eqref{Ppbi:ter}. This allows to characterize any symmetric Nash equilibrium if there exists at least one.
\end{remark}

%
%  \[ \hat H^S_t(x)=\frac12 \sigma^2\gamma \big( (x+\frac{\tilde Z^S_t}2+Q_t)^2 -\frac12 ({\tilde Z^S_t}+Q_t)^2 \big), \]
%                                     \[ \hat H^i_t(x)= H^i(x+\frac{\tilde Z^i_t}2,Q_t) - \frac12 H^i(\tilde Z^i_t,Q_t).  \]
% \\

\subsection{The main result}
 \noindent Similarly to the one exchange problem studied previously, we introduce the HJB equation 
 \begin{equation} \label{edp_v^i_hjb}
 v\big|_{t=T} = -1
 ~\mbox{and}~
 \partial_t v(t,q) -\eta v(t,q) \hat F(t,q,v(t,q),v(t,q+1),v(t,q-1))\;=\;
0, 
\end{equation}
with
\[  \hat F(t,q,y,y^+,y^-)=\sup_{\zeta^S}  F^S(t,q,\zeta^S) +  \sup_{\zeta}F^{0}(t,q,y,y^+,\zeta) \mathbf 1_{q<\overline q}+  \sup_{\zeta}F^{0}(t,q,y,y^-,\zeta) \mathbf 1_{q>-\overline q}   ,  \]
and
 \begin{eqnarray*} 
 F^S(t,q,z)
 &=& 
 -\hat H^S(q,\zeta^{S,0}(q),z)-\frac\eta2 \sigma^2 |z|^2,  
 \\
 F^{0}(t,q,y,y',z)
 &=& 
 - \lambda\big(\hat\delta(z+(N-1)\widetilde{\zeta^{0}}(y,y'))\big)
   \Big(\frac{y'}y\frac{e^{\eta (z-\frac cN) }}{\eta}-\frac1\eta- \frac{\sigma}{k+\sigma\gamma}\Big)
 \\
&& 
- \lambda\big(\hat\delta(N \widetilde{\zeta^{0}}(y,y'))\big) \frac{(N-1)\sigma}{N(k+\sigma\gamma)}, 
\\
\hat H^S(q,\tilde z,z)
&=&
\frac12 \sigma^2\gamma \Big[ (z+(N-1)\tilde z+q)^2 -\frac{N-1}N (N\tilde z+q)^2 \Big],
\\
\zeta^{S,0}(q)
 &=&
 -\frac{\gamma}{\eta+N\gamma}q,
 ~~
\widetilde{\zeta^{0}} (y,y')
 \;=\;
\hat{\zeta}(y,y')+\frac{1-N}N c.
\end{eqnarray*}
Hence, by denoting $\widehat{\zeta}^{S},\, \widehat{\zeta}^N$ the optimizers of $F^{S}$ and $F^0$ respectively, we get 
\[\widehat{\zeta}^{S}(q)= \zeta^{S,0}(q),\; \widehat{\zeta}^N(y,y')=\widetilde{\zeta^{0}}(y,y').   \]
 
\noindent Consequently, the HJB equation \eqref{edp_v^i_hjb} reduces to 
 \begin{equation} \label{edp_v^i}
 \partial_t v(t,q) 
 +  
 \frac{\gamma \eta^2 \sigma^2}{2N(N\gamma + \eta)} q^2 v(t,q) 
 - \hat C_N v(t,q) \Big[\1_{\{q>-\bar q\}} \big(\frac{v(t,q)}{v(t,q-1)}\big)^{\frac{Nk}{\sigma \eta}} 
                        + \1_{\{q<\bar q\}} \big(\frac{v(t,q)}{v(t,q+1)}\big)^{\frac{Nk}{\sigma \eta}} 
                \Big]
\;=\;
0,
\end{equation}
with boundary condition $v\big|_{t=T}=-1$, where 
 \[
 \hat C_N
 = 
 C_0e^{\frac{(N-1)k}{\sigma\eta}}
 \;\frac{\sigma\gamma  +\frac1N (\sigma\eta+k)}{\sigma\gamma  + (\sigma\eta+k)}.\]

\noindent We now set $u = (-v)^{-\frac{kN}{\sigma \eta}}$. By direct substitution, we obtain the following linear equation 
\begin{equation} \label{linearPDE:Nex}
u\big|_{t=T}=1
~\mbox{and}~
\partial_t u(t,q) - F_{C_N,C_N'}(q,u(t,q),u(t,q+1),u(t,q-1))
=
0 , \quad t \in [0, T),
\end{equation}
with 
 \begin{eqnarray*}
C_N \;=\; \frac{k \gamma \eta \sigma}{2(N\gamma + \eta)}
 &\mbox{and}&
C'_N = \hat C_N \frac{kN}{\sigma \eta}.
 \end{eqnarray*}
Similarly to Section 4, we deduce that
 \begin{eqnarray*}
 v(t,q)
 =
 -\big(u(t,q)\big)^{-\frac{\sigma\eta}{kN}}
 &\mbox{where}&
 u(t,q)=\mathbf{b}_q\!\cdot\! e^{(T-t)\mathbf{B}_N}\mathbf{1},
 \end{eqnarray*}
 and 
  \begin{eqnarray*}
\mathbf{B}_N
 =
 {\tiny
 \left(
 \begin{array}{ccccc}
 -C_N{\bar q}^2    &C'_N           &                   &                   & 
 \\
 \ddots & \ddots & \ddots  &                   & 
 \\
                   &C'_N           & -C_Nq^2    &C'_N            &  
 \\
                   &                   & \ddots & \ddots & \ddots  
 \\
                   &                   &                   &   C'_N       &    -C_N{\bar q}^2  
 \end{array}
 \right)}
 \leftarrow~q\mbox{-th line,} 
 \end{eqnarray*}
Direct calculations reported in Appendix \ref{Appendix-calcul u} provide an other form for the function $u$: 
 \begin{equation}\label{calcul:u:N}
 \begin{array}{rcl}
 u(t,q)
 &=&
 \displaystyle
 \sum_{p\geq 0}  \frac{ [C'_N(T-t)]^p}{p!} \sum_{j\ge 0}  \frac{ [C'_N(T-t)]^j}{j!} e^{-C_N(T-t) (q+j-p)^2}
 \1_{\{|q+j-p|\le\bar q\}}.
\end{array}
 \end{equation}

\noindent The following result establishes the existence of a unique symmetric Nash equilibrium which is moreover Markovian. The proof is postponed to Appendix \ref{appendix:nash}.
\begin{theorem}\label{thmnash}
There is a unique symmetric Nash equilibrium $\xi^0\in\Sc^N$ defined by                              
\begin{equation}\label{optimal:contract:2ex} 
\xi^0
= 
\frac{\widehat{Y_0}}N
+\int_0^T \zeta^0_r d\chi_r
               +\frac{\gamma \sigma^2}{2N} (N \zeta_r^{S,0}+Q_r)^2dr
               -\frac1N\sum_{i=a,b} H^i(N\zeta_r^{i,0},Q_r) dr,
\end{equation}
where, for $i\in\{a,b\}$ and $(\varepsilon_b,\varepsilon_a)=(-1,1)$,
$$
\zeta^{S,0}_r
=
-\frac{\gamma}{\eta+N\gamma}Q_r,
~~
\zeta^{i,0}_r
=
\frac{c}N+\frac 1\eta\Big(\log\big(\frac{v(r,Q_r)}{v(r,Q_r-\varepsilon_i)}\big) 
+\log\big(1-\frac{\sigma^2\gamma\eta}{(\sigma\eta+k)(\sigma\gamma+k)} \big) \Big).
$$
In particular, this unique symmetric Nash equilibrium is Markovian.
\end{theorem}

\begin{remark}
There exists infinitely many (non symmetric) Nash equilibria. For instance, by the contract $(\xi^{\rm e})_{{\rm e}\leq N}$ defined by
\[\
\xi^{\rm e}
= 
Y_0^j
+\int_0^T \zeta^0_r d\chi_r
               +\frac{\gamma \sigma^2}{N}(N\zeta_r^{S,0}+Q_r)^2dr
               -\frac1N\sum_{i=a,b} H^i(N\zeta_r^{i,0},Q_r) dr,
\]
is a (non symmetric) Nash equilibrium for any $Y_0^j$ satisfying $\sum_{j=1}^N Y_0^j=\hat Y_0$.

\end{remark}

\begin{remark}
The symmetry of the problem allows us to find a natural candidate for the equilibrium and by using a verification procedure (see the first step of the proof) we prove that it is indeed a (symmetric) Nash equilibrium. This follows from the fact that the integro-differential equation under consideration admits a smooth solution. If now one wants to extend this study to an heterogeneous oligopoly of exchanges hiring one market maker, the solution will be strongly linked to a system of fully coupled HJB equation has explained in \cite{mastrolia2018principal}. However, the existence of a smooth solution for such a system is not clear.

\end{remark}

\subsection{Economic insights}

Notice that the total compensation $N\xi^0$ obtained by the market maker in the $N-$symmetric exchanges situation differs from the optimal contract $\hat \xi$ of the one-exchange situation in \eqref{optimal:contract:2ex}. Hence, our result is not trivial and worth of interest even in the simple symmetric exchanges setting in symmetric Nash equilibrium.

Notice that we also have a similar representation of $u$ as in Proposition \ref{prop:Representation}: \begin{eqnarray}\label{eq:u:nash}
 u(t,q)
 &=&
 \E\Big[ e^{\int_t^T (-C_N(Q^{t,q}_s)^2+\overline{\lambda}_s+\underline{\lambda}_s)ds}\Big],
 \end{eqnarray}
where $Q^{t,q}_s=q+\int_t^s d(\overline{N}_u-\underline{N}_u)$, and $(\overline{N},\underline{N})$ is a two-dimensional point process with intensity $(\overline{\lambda}_s, \underline{\lambda}_s) = C'_N (\1_{\{Q_{s-} <\bar q\}}, \1_{\{Q_{s-} >-\bar q\}})$. By using the same arguments than those in Section \ref{discussion} (and based on the asymptotic expansion in \cite{avellaneda2008high,gueant2013dealing}) we note that $N\zeta^{b}_t\sim -\frac{C_N}kQ_t$ and $N\zeta^{a}_t\sim \frac{C_N}kQ_t$. Again, when the inventory is highly positive, the exchanges provide incentives to the market-maker to attract buy market orders and discourage additional sell market orders, and vice versa for a negative inventory. As $C_N$ is decreasing with respect to $N$, this effect is reduced when the number of platforms increases. Consequently, from the market maker point of view, the observed spread is reduced when the number of exchange grows.

Note now that when $N$ becomes large, $N\zeta^{S,0}_r \sim -Q_r$. In other words, for a large number of platforms, the inventory risk is transferred to the oligopoly of exchanges.

  \appendix
\section{Appendix}
\subsection{Predictable representation}

The following result is probably well-known, we report it for completeness as we could not find a precise reference.

\begin{lemma} \label{lemmaRepresentationMart} Let $(\Omega, {\cal F}, \P, \F)$ be a filtered probability space where $\F = \F^W \lor \F^N$ is the right continuous completed filtration of a Brownian motion $W$ and a $d$-dimensional integrable point process $N = (N^1, \cdots, N^d)$ with compensator $A = (A^1, \cdots, A^d)$. Then, for any $\F-$martingale $X$ there exists a predictable process $Z=(Z^W, Z^1, \cdots, Z^d)$ such that
$$ X_t = X_0 + \int_0^t Z^W_s dW_s + \sum_{i = 1}^d \int_0^t Z_s^i (dN_s^i - dA_s^i). $$
\end{lemma}

\begin{proof} For sake of simplicity, we take $d = 1$. Let $\P$ be a solution of the martingale problem associated to $M_t= N_t - A_t$ and $W_t$. By Theorem III.4.29 in \cite{jacod2013limit}, to prove Lemma \ref{lemmaRepresentationMart} we need to establish the uniqueness of $\P$.

We denote by $\P^W$ the law $\P$ conditional on $W$. We first show that $M$ is still a martingale under $\P^W$. To do so we consider 
$B_s \in {\cal F}_s$ and want to prove that 
$$\E^{\P^W}\big[{\1_{B_s}} (M_t - M_s)\big] = 0,
$$
for $0 \leq s \leq t \leq T$. Let $C \in {\cal F}_T^W$. We aim at showing that 
$$\E\Big[1_C\E^{\P^W}\big[{\1_{B_s}} (M_t - M_s)\big] \Big]= \E\big[{\1_C \1_{B_s}} (M_t - M_s)\big]=0.$$
By the martingale representation theorem for Brownian martingales, we can write $ \1_C = \alpha_s + \int_s^T \phi_u dW_u,$
where $\alpha_s = \E[\1_C| \mathcal{F}^W_s]$ and $(\phi_u)_{u \geq 0}$ is $\F^W$ predictable process. Using the martingale property of $M$, we obtain
$$\E\big[\alpha_s \1_{B_s} (M_t - M_s)\big]=0.$$
Then $W$ and $M$ being {orthogonal martingales}, we deduce
$$\E\Big[\int_s^T \phi_u dW_u \1_{B_s} (M_t - M_s) \Big] = 0,  $$
Consequently, using Theorem III.1.21 in \cite{jacod2013limit}, $\P^W$ is the unique probability measure such that $M$ is an $\F$-martingale conditional on $W$. Finally, by integration, the uniqueness of $\P^W$ implies that of $\P$.
\end{proof}

\subsection{Exchange Hamiltonian maximization}

The following result follows from (tedious) direct calculations.

\begin{lemma} \label{lemmaAnnex}
For all $v_1, v_2 < 0$, define 
$$ 
\varphi(z) 
:=
A e^{-k\frac{\Delta(z)+c}{\sigma}} 
\Big[v_1 e^{\eta(z-c)} 
        - v_2 \Big(\frac{\eta}{\gamma}\big(1-e^{-\gamma(z+\Delta(z))}\big) + 1\Big)
\Big],
~~z \in \R,
$$
with $\Delta(z) := (-\delta_\infty) \lor \Delta^0(z)\land \delta_\infty$ and $\Delta^0(z):=-z + \log(1+\frac{\sigma \gamma}{k})^{\frac{1}{\gamma}}$, for some parameter $\delta_\infty$ satisfying
\begin{eqnarray} \label{condition:v1v2} 
\delta_\infty 
\ge C_\infty + \big|\log\Big(\frac{v_2}{v_1}\Big)^{\frac{1}{\eta}}\big|,
&\mbox{with}&
C_\infty 
:=  
c + \log\frac{(1+\frac{\sigma \gamma}{k})^{\frac{1}{\eta}+\frac{1}{\gamma}} }
                   {\big(1 - \frac{\sigma^2 \gamma \eta}{(k + \sigma \gamma)(k + \sigma \eta)} \big)^{\frac{1}{\eta}}}.
\end{eqnarray}
Then, $\varphi$ has a maximum point  $z^\star$ given by:
\begin{eqnarray*}
z^\star 
= 
c + {\frac{1}{\eta}}\log\Big(\frac{\frac{v_2}{v_1}}
                          {1 - \frac{\sigma^2 \gamma \eta}{(k + \sigma \gamma)(k + \sigma \eta)}}
             \Big)
&\mbox{with}&
\varphi(z^\star) = - C v_2 \Big(\frac{v_2}{v_1}\Big)^{\frac{k}{\sigma \eta}} ,
~~
\big|\Delta^0(z^\star)\big|\le\delta_\infty,
\end{eqnarray*}
and $
C 
= 
A \frac{\sigma \eta}{k}
\big(1+\frac{\sigma \gamma}{k}\big)
^{-\frac{k}{\sigma \gamma}} 
\big(1 - \frac{\sigma^2 \gamma \eta}{(k + \sigma \gamma)(k + \sigma \eta)}\big)
   ^{1+\frac{k}{\sigma \eta}}.
$
\end{lemma}

\subsection{Justification of \eqref{calcul:u} and \eqref{calcul:u:N}}\label{Appendix-calcul u}

Denote by $\mathbf D$ and $\mathbf J$ the matrices defined by the entries $\mathbf D^{q,r}= q^2\mathbf 1_{q=r}$ and $\mathbf J^{q,r}= \1_{r=q+1}+\1_{r=q+1}$, $-\overline q\leq p,r\leq \overline q$. Notice that the calculations reported in \eqref{calcul:u} and \eqref{calcul:u:N} reduce to:
 \begin{eqnarray*}
 U(q)
 :=
 \mathbf{b}_q
 \cdot
 \sum_{|\ell|\le\bar q} e^{\alpha\mathbf{J} -\beta \mathbf D}\,\mathbf{b}_\ell,
 ~|q|\le\bar q,
 &\mbox{for}&
 -\overline q\leq q\leq \overline q,
 \end{eqnarray*}
We first compute that
 $$
 e^{\alpha\mathbf{J} -\beta \mathbf D}
 =
 \sum_{k\geq0}\frac{(\alpha \mathbf J-\beta \mathbf D)^k}{k!}
 =
 \sum_{k\geq0} \frac{1}{k!}\sum_{j=0}^k \binom kj \alpha^j (-\beta)^{k-j}\ell^{2(k-j)}\mathbf J^j 
 $$
As $\mathbf J^j\cdot\mathbf{b}_\ell=\sum_{p=0}^j  \binom jp \mathbf b_{\ell-j+2p}$, and $\mathbf{b}_q\cdot\mathbf{b}_{q'}=\1_{\{q=q'\}}$, this provides
  \begin{eqnarray*}
 \mathbf{b}_q \cdot e^{\alpha\mathbf{J} -\beta \mathbf D}\mathbf{b}_\ell
 &=&
 \sum_{k\geq0} \frac{1}{k!}\sum_{j=0}^k \binom kj \alpha^j (-\beta)^{k-j}\ell^{2(k-j)}\sum_{p=0}^j  \binom jp \1_{\{\ell-j+2p=q\}}
 \\
&=&
\sum_{k\geq0}\sum_{j=0}^k \sum_{p=0}^j \frac{ \alpha^j (-\beta \ell^2)^{k-j}}{p!(k-j)!(j-p)!} \1_{\{\ell-j+2p=q\}}
 \\
 &=&
\sum_{p\ge 0}\sum_{j\ge p} \frac{ \alpha^j}{p!(j-p)!} \;e^{-\beta\ell^2}\1_{\{\ell=q+j-2p\}}
 \;=\;
 \sum_{p\ge 0}\sum_{j\ge 0} \frac{ \alpha^{j+p}}{p! j!} \;e^{-\beta\ell^2}\1_{\{\ell=q+j-p\}},
  \end{eqnarray*}
and we finally conclude that
 \begin{eqnarray*}
 U(q)
 \;=\;
 \sum_{p\ge 0} \sum_{j\ge 0} \frac{ \alpha^{j+p}}{p! j!} \;e^{-\beta(q+j-p)^2}
                                              \1_{\{|q+j-p|\le\bar q\}},
 &\mbox{for}&
 |q|\le\bar q.
 \end{eqnarray*}

\subsection{Dynamic programming principle and representation}
\label{sect:DPP}

\noindent For all $\mathbb F$-stopping time $\tau$ with values in $[t,T]$ and for any $\mu\in \mathcal A_\tau$, we define\footnote{From \eqref{proba_change}, notice that for any $\delta \in \Ac$, the conditional expectation $\E^\delta_\tau$ depends only on the restriction of $\delta$ on $[\tau, T]$. Hence $\E^\mu_\tau$ is defined without ambiguity for $\mu \in \Ac_\tau$.}
 \begin{eqnarray*}
 J_T(\tau,\mu)=\mathbb E_\tau^{\mu}\left[-e^{-\gamma\int_\tau^{T}(\mu^a_u dN^a_u + \mu^b_u dN^b_u + Q_u dS_u)}  
  e^{-\gamma\xi} \right],
 &\mbox{and}&
 \mathcal J_{\tau,T}=\left( J_T(\tau,\mu)\right)_{\mu\in \mathcal A_\tau},
 \end{eqnarray*}
 where $\mathcal A_{\tau}$ denotes the restriction of $\mathcal A$ to controls on $[\tau,T]$.
The continuation utility of the market maker is defined for all $\mathbb F$-stopping time $\tau$ by
 \begin{eqnarray*}
 V_\tau
 &=&
 \underset{\mu \in {\cal A_\tau}}{\textnormal{ess} \sup}\, J_T(\tau,\mu).
\end{eqnarray*}

\begin{lemma}\label{familleJ}
Let $\tau$ be an $\mathbb F$-stopping time with values in $[t,T]$. Then, there exists a non-decreasing sequence $(\mu^n)_{n\in \mathbb N}$ in $\mathcal A_\tau$ such that $V_\tau=\underset{n\rightarrow +\infty}{\lim}\!\!\!\!\uparrow J_T(\tau,\mu^n)$.
\end{lemma}

\begin{proof}
For $\mu$ and $\mu'$ in $\mathcal A_\tau$, define
$\hat{\mu}=\mu  \1_{\{J_{T}(\tau,\mu)\geq J_{T}(\tau,\mu') \}} +\mu'  \1_{\{J_{T}(\tau,\mu)< J_{T}(\tau,\mu')\}}.$
Then $\hat{\mu}\in \mathcal A_{\tau}$ and by definition of $\hat{\mu}$, we have $J_T(\tau,\hat{\mu})\geq \max\left( J_T(\tau,\mu),J_T(\tau,\mu')\right).$ This shows that $\mathcal J_{\tau,T}$ is directly upwards, and the required result folows from \cite[Proposition VI.I.I p121]{neveu1972martingales}.
\end{proof}

\begin{lemma}\label{lemme_DPP}
 Let $t\in[0,T]$ and $\tau$ {be} an $\mathbb F$-stopping time  with values in $[t,T]$. Then,
 \begin{eqnarray*}
 V_t 
 &=& 
 \underset{\delta \in {\cal A}}{\mbox{\rm ess} \sup} \,
 \E^{\delta}_t\Big[  e^{-\gamma \int_t^{\tau}(\delta_u \cdot dN_u + Q_u dS_u) } V_\tau
                     \Big].
 \end{eqnarray*}
\end{lemma}

\begin{proof}  We follow the same argument as in \cite[Proof of Proposition 6.2]{cvitanic1993hedging}. Denote $\tilde V_t$ the right hand side of the required equality. First, by the tower property, 
 \begin{eqnarray*}
 V_t  
 &=&
 \underset{\delta \in {\cal A}}{\text{ess} \sup} \,\E^{\delta}_t\Big[  e^{-\gamma \int_t^{\tau}(\delta_u\cdot dN_u + Q_u dS_u)} \mathbb E^\delta_\tau\left[ -e^{-\gamma\left( \int_{\tau}^T(\delta_u\cdot dN_u + Q_u dS_u) +\xi\right)}\right]\Big].
\end{eqnarray*}
For all $\delta\in \mathcal A$, the quotient $\frac{L^\delta_T}{L^\delta_\tau}$ does not depend on the values of $\delta$ before time $\tau$. Then,
 \begin{eqnarray*}
 \mathbb E^\delta_\tau\left[- e^{-\gamma\left( \int_{\tau}^T(\delta_u\cdot dN_u + Q_u dS_u) +\xi\right)}\right]
 &\!\!=&\!\!
 \mathbb E^0_\tau\left[- \frac{L^\delta_T}{L^\delta_\tau}e^{-\gamma\left( \int_{\tau}^T(\delta_u\cdot dN_u + Q_u dS_u) +\xi\right)}\right]
 \\
 &\!\!\le&\!\!
 \underset{\mu \in {\cal A}_\tau}{\text{ess} \sup}\,  \mathbb E^\mu_\tau\left[-e^{-\gamma\left( \int_{\tau}^T(\mu_u\cdot dN_u + Q_u dS_u) +\xi\right)}\right]
 =
 V_\tau,
\end{eqnarray*}
which implies that $V_t\le\tilde V_t$.

We next prove the reverse inequality. Let $\delta\in \mathcal A$ and $\mu\in \mathcal A_\tau$. We define $(\delta\otimes_\tau \mu)_u= \delta_u\mathbf 1_{0\leq u<\tau}+\mu_u \mathbf 1_{\tau\leq u\leq T}$. Then $\delta\otimes_\tau \mu\in \mathcal A$ and
 \begin{eqnarray}
 \nonumber 
 V_t
 &\ge& 
 \mathbb E^{\delta\otimes_\tau \mu}_t  \left[ -e^{-\gamma\big( \int_t^{\tau}(\delta_u\cdot dN_u + Q_u dS_u)+\int_{\tau}^T(\mu_u\cdot dN_u + Q_u dS_u) \big)} e^{-\gamma\xi} \right]
 \\
 \label{ppd:ineg1}
 &=&
 \mathbb E^{\delta\otimes_\tau \mu}_t  \left[ e^{-\gamma\int_t^{\tau}(\delta_u\cdot dN_u + Q_u dS_u)} \mathbb E^{\delta\otimes_\tau\mu}_\tau\Big[-e^{-\gamma\int_{\tau}^T(\mu_u\cdot dN_u + Q_u dS_u) } e^{-\gamma\xi} \Big]\right].
 \end{eqnarray}
From Bayes' Formula and by noticing that $\frac{L_T^{\delta\otimes_\tau\mu}}{L^{\delta\otimes_\tau\mu}_\tau}=\frac{L_T^{\mu}}{L^{\mu}_\tau}$, we get
 $$
 \mathbb E^{\delta\otimes_\tau\mu}_\tau\left[-e^{-\gamma \int_{\tau}^T(\mu_u\cdot dN_u + Q_u dS_u)} e^{-\gamma\xi} \right]
 =
 \mathbb E^0_\tau\left[\frac{L_T^{\delta\otimes_\tau\mu}}{L^{\delta\otimes_\tau\mu}_\tau}\left(-e^{-\gamma \int_{\tau}^T(\mu_u\cdot dN_u + Q_u dS_u)} e^{-\gamma\xi} \right)\right]
 =
 J_T(\tau,\mu).
$$
Thus, Inequality \eqref{ppd:ineg1} becomes $V_t
 \ge
 \mathbb E^{\delta\otimes_\tau \mu}_t  \left[ e^{-\gamma \int_t^{\tau}(\delta_u\cdot dN_u + Q_u dS_u)} J_T(\tau,\mu)\right]$, and by using again Bayes' Formula and by noticing that $\frac{L_\tau^{\delta\otimes_\tau \mu}}{L_t^{\delta\otimes_\tau \mu}}=\frac{L_\tau^{\delta}}{L_t^{\delta}}$, we have
 \begin{eqnarray*}
 V_t
 &\ge&
 \frac{\mathbb E^0_t  \left[ L_T^{\delta\otimes_\tau \mu}e^{-\gamma \int_t^{\tau}(\delta_u\cdot dN_u + Q_u dS_u)} J_T(\tau,\mu)\right]}{ L_t^{\delta\otimes_\tau \mu}}
 \\
 &=&
 \mathbb E^0_t  \left[ \mathbb E^0_\tau\Big[\dfrac{L_T^{\delta\otimes_\tau \mu}}{L^{\delta\otimes_\tau \mu}_\tau}\dfrac{L_\tau^{\delta\otimes_\tau \mu}}{L_t^{\delta\otimes_\tau \mu}}e^{-\gamma \int_t^{\tau}(\delta_u\cdot dN_u + Q_u dS_u)} J_T(\tau,\mu)\Big] \right]
 \\
 &=&
 \mathbb E^0_t  \left[ \mathbb E^0_\tau\Big[\dfrac{L_T^{\delta\otimes_\tau \mu}}{L^{\delta\otimes_\tau \mu}_\tau}\Big]\dfrac{L_\tau^{\delta\otimes_\tau \mu}}{L_t^{\delta\otimes_\tau \mu}}e^{-\gamma \int_t^{\tau}(\delta_u\cdot dN_u + Q_u dS_u)} J_T(\tau,\mu)\right]
 \\
 &=&
 \mathbb E^0_t  \left[\dfrac{L_\tau^{\delta\otimes_\tau \mu}}{L_t^{\delta\otimes_\tau \mu}}e^{-\gamma \int_t^{\tau}(\delta_u\cdot dN_u + Q_u dS_u)} J_T(\tau,\mu)\right]
 \;=\;
 \mathbb E_t^\delta  \left[e^{-\gamma \int_t^{\tau}(\delta_u\cdot dN_u + Q_u dS_u)} J_T(\tau,\mu)\right].
\end{eqnarray*}
Since the previous inequality holds for all $\mu\in \mathcal A_{\tau}$ we deduce from monotone convergence Theorem together with Lemma \ref{familleJ} that there exists a sequence $(\mu^n)_{n\in \mathbb N}$ in $\mathcal A_\tau$ such that 
 \begin{eqnarray*} 
 V_t
 &\ge&
 \underset{n\rightarrow +\infty}{\lim}\!\!\!\!\uparrow  \mathbb E_t^\delta  \left[e^{-\gamma\int_t^{\tau}(\delta_u\cdot dN_u + Q_u dS_u)} J_T(\tau,\mu^n)\right]
 \\
 &=&  \mathbb E_t^\delta  \left[e^{-\gamma \int_t^{\tau}(\delta_u\cdot dN_u + Q_u dS_u)} \underset{n\rightarrow +\infty}{\lim}\!\!\!\!\uparrow  J_T(\tau,\mu^n)\right]
  \;=\; \mathbb E_t^\delta  \left[e^{-\gamma \int_t^{\tau}(\delta_u\cdot dN_u + Q_u dS_u)}V_\tau\right],
\end{eqnarray*}
thus concluding the proof.
\end{proof}

\vspace{5mm}

\noindent {\bf Proof of Theorem \ref{thm:Agent} (i)} We proceed in several steps.
\\
{\it Step 1.} For $\delta \in \Ac$, it follows from the dynamic programming principle of Lemma \ref{lemme_DPP} that the process ${U}^\delta_t :=
 V_t \,e^{-\gamma \int_0^t(\delta^a_u dN^a_u + \delta^b_u dN^b_u + Q_u dS_u)}$, $t\in[0,T]$,
defines a $\P^{\delta}$-supermartingale\footnote{{Note that $\E^\delta[U^\delta_T] = J_T(0, \delta) > - \infty$ using H\"older inequality together with \eqref{condition:xi} and the uniform boundedness of the intensities of $N^a$ and $N^b$. Hence the process $U^\delta$ is integrable}.} for all $\delta\in\cal A$. By standard analysis\footnote{In view of the class of contracts considered we know that the principal proposes a contract such that there exists at least one optimal bid-ask spread for the agent denoted by $\tilde \delta$. Hence, $U_t^{\tilde \delta}$ is a $\mathbb P^{\tilde \delta}$-martingale and according to Doob regularization result, we know that we can find a c\`adl\`ag version of $U_t^{\tilde \delta}$ under $\mathbb P^{\tilde \delta}$. Thus $V_t$ admits a c\`adl\`ag version under $\mathbb P^{\tilde \delta}$, and since all the measure $\mathbb P^\delta$ for $\delta\in \mathcal A$ are equivalent, $U_t^\delta$ admits a c\`adl\`ag version.}, we may then consider it in its c\`adl\`ag version (by taking right limits along rationals). By the Doob-Meyer decomposition, we write
 \begin{eqnarray}\label{DoobMeyer}  
 U^\delta_t 
 &=&  M^\delta_t - A^{\delta,c}_t - A^{\delta,d}_t, 
\end{eqnarray}
where $ M^\delta$ is a $\P^{\delta}$-martingale and $A^\delta=A^{\delta,c} + A^{\delta,d}$ is an integrable non-decreasing predictable process such that $A^{\delta,c}_0=A^{\delta,d}_0 = 0$, with pathwise continuous component $A^{\delta,c}$, and a piecewise constant predictable process $A^{\delta, d}$.
By the martingale representation theorem under $\P^\delta$, see Lemma \ref{lemmaRepresentationMart}, we have
 \begin{eqnarray}\label{tildeMdelta}  
 M^\delta_t 
 &=& 
 V_0 
 + \int_0^t  \widetilde{Z}^{\delta}_r.d\chi_r 
 -\int_0^t\widetilde{Z}^{\delta,a}_r \lambda(\delta^a_r)\1_{\{Q_r >-\bar q\}} dr  
 -\int_0^t\widetilde{Z}^{\delta,b}_r \lambda(\delta^b_r)\1_{\{Q_r <\bar q\}} dr, 
 \end{eqnarray}
predictable process $\widetilde{Z}^\delta = (\widetilde{Z}^{\delta,S}, \widetilde{Z}^{\delta,a}, \widetilde{Z}^{\delta,b})$, where we recall that $\chi = (S, N^a, N^b)$. 
\\

\noindent {\it Step 2.} We show that $V$ is a negative process. In fact, thanks to the uniform boundedness of $\delta \in \Ac$, we show that
\begin{equation} \label{ineq_alpha}
 \frac{L^\delta_T}{L^\delta_t} 
 \ge
 \alpha_{t,T} 
 = e^{-\frac{k \delta_\infty}{\sigma} (N^a_T - N^a_t+ N^b_T - N^b_t) -  2Ae^{-\frac{kc}{\sigma}}(e^{\frac{k \delta_\infty}{ \sigma}} + 1) (T-t)}
 >0,
\end{equation}
which implies that $V_t 
\le 
\E^0\left[ - \alpha_{t,T} e^{- \gamma \big(\delta_\infty (N^a_T - N^a_t + N^b_T - N^b_t) + \int_t^T Q_u dS_u\big) } e^{- \gamma \xi} \right]
< 0$.\\

\noindent {\it Step 3.} Let $Y$ be the process defined by $V_t =- e^{-\gamma Y_t}$ for all $t\in[0,T]$. As $A^{\delta,d}$ is a predictable point process and the jumps of $(N^a, N^b)$ are totally inaccessible stopping times under $\P^0$, we have $[N^a, A^{\delta,d}] = 0$ and $[N^b, A^{\delta,d}] =0$ a.s, see Proposition I.2.24 in \cite{jacod2013limit}. Using It\^o's formula, we obtain from \eqref{DoobMeyer} and \eqref{tildeMdelta} that
 \begin{eqnarray*}
 Y_T = \xi,
 &\mbox{and}&
 dY_t =  Z^a_t dN^a_t + Z^b_t dN^b_t + Z^S_t dS_t - dI_t - d\widetilde A^d_t,
 \end{eqnarray*}
where $Z^a, Z^b,Z^S,I,\widetilde A^d$ are independent of $\delta$, as they may be expressed as {$Z^i_t dN^i_t =d[ Y,  N^i]_t$}, $i\in\{a,b\}$, $Z^S_t \sigma^2 dt = d\langle Y_t,S_t\rangle_t$, $\widetilde A^d$ the predictable pure jumps of $Y$. {Moreover, It\^o's Formula yields } 
$$ Z_t^a = -\frac{1}{\gamma} \log(1 + \frac{\widetilde{Z}^{\delta,a}_t}{ {U}^\delta_{t-}}) - \delta^a_t ,
 \quad
 Z_t^b = -\frac{1}{\gamma} \log(1 + \frac{\widetilde{Z}^{\delta,b}_t}{ {U}^\delta_{t-}}) - \delta^b_t ,
\quad
  Z_t^S = -\frac{\widetilde{Z}^{\delta,b}_t}{\gamma{U}^\delta_{t-}} - Q_{t-} ,
 $$
and 
 \begin{eqnarray*}
 I_t = \int_0^t\left( \overline{h}(\delta_r, Z_r, Q_r) dr - \frac{1}{\gamma U^\delta_{r}} dA^{\delta, c}_r \right),
 &&
 \widetilde A^{d}_t 
 =
\frac{1}{\gamma} \sum_{s\le t} \log{\Big(1-\frac{\Delta A^{\delta,d}_t}{U^\delta_{t-}}\Big)},
 \end{eqnarray*}
with $\overline{h}(\delta, z, q)=h(\delta, z, q)-\frac12\gamma\sigma^2 (z^s)^2$. In particular, the last relation between $\widetilde A^d$ and $A^{\delta,d}$ shows that $\Delta a_t=\frac{-\Delta A^{\delta,d}_t}{U^\delta_{t-}}\ge 0$ is independent of $\delta \in \Ac$; recall that $U^\delta<0$.

In the subsequent steps, we argue  that $Z = (Z^S, Z^a, Z^b) \in \Zc$, and 
 \begin{eqnarray}\label{goal:representation}
 A^{\delta,d}_t 
 =
 -\sum_{s\le t} U^\delta_{s-}\Delta a_s
 =0,~~\mbox{(so that $\widetilde A^{d}_t =0$)}, 
 &\mbox{and}&
 I_t=\int_0^t \overline{H}(Z_r,Q_r) dr,
 ~t\in[0,T],
 \end{eqnarray}
 where $\overline{H}(z, q)=H( z, q)-\frac12\gamma\sigma^2 (z^s)^2$.\\
 
\noindent {\it Step 4.} Since $V_T=-1$, notice that
 \begin{eqnarray}
 0
 \;=\;
 \sup_{\delta \in \Ac} ~\E^\delta[ U^\delta_T] -V_0
 &=&
 \sup_{\delta \in \Ac} ~\E^\delta[ U^\delta_T- M^\delta_T ]
 \nonumber\\
 &=&
 \gamma
 \sup_{\delta \in \Ac} ~
 \E^0\Big[ L^\delta_T \int_0^T  U^\delta_{r-} \big(dI_r-\overline{h}(\delta_r, Z_r, Q_r)dr
                                                                           +{\frac{da_r}\gamma}
                                                                    \big)
          \Big].
 \label{presquefini}
 \end{eqnarray}
Moreover, since the controls are uniformly bounded, we have
\begin{equation} \label{ineq_beta}
 U^\delta_t 
 \le  
 -\beta_t 
 \;:=\;
 V_t e^{- \gamma \delta_\infty (N_t^a - N_0^a + N_t^b - N_0^b) - \gamma \int_0^t Q_r dS_r}<0.
 \end{equation}
Then, since $A^{\delta,d}\ge 0$, $U^\delta\le 0$, and $dI_t-\overline{h}(\delta_t, Z_t, Q_t)\ge 0$, it follows from \eqref{presquefini}  together with the inequalities \eqref{ineq_alpha} and \eqref{ineq_beta} that
 \begin{eqnarray*}
 0 
 &\le&
 \sup_{\delta \in \Ac}
 \E^0\Big[ \alpha_{0,T} \int_0^T  -\beta_{r-} \big(dI_r - \overline{h}(\delta_r, Z_r, Q_r)dr
                                                                           +\frac{da_r}{\gamma}
                                                                    \big)
          \Big ]\\
 &=&
 - \E^0\Big[ \alpha_{0,T} \int_0^T  \beta_{r-} \big(dI_r-\overline{H}(Z_r, Q_r)dr
                                                                           +\frac{da_r}{\gamma}
                                                                    \big)
          \Big ].
 \end{eqnarray*}
As $\alpha_{0,T} \int_0^T  \beta_{r-} \big(dI_r-\overline{H}(Z_r, Q_r)dr)\ge 0$ and $\alpha_{0,T} \int_0^T  \beta_{r} da_r\ge 0$, 
this implies \eqref{goal:representation}.
\\

\noindent {\it Step 5.} We now prove that $Z \in \Zc$ by showing that 
\begin{equation} \label{step5}
\sup_{\delta \in \Ac} ~ \sup_{t \in [0,T]} \E^\delta[e^{- \gamma (p+1) Y_t}] < \infty \quad \text{for some $p>0$}.
\end{equation}
Using H\"older inequality together with Condition \eqref{condition:xi} and the boundedness of the intensities of $N^a$ and $N^b$, we have that $\sup_{\delta \in \Ac} \E^\delta[|U^\delta_T|^{p' + 1}] < \infty$ for some $p'>0$. Hence
$$\sup_{\delta \in \Ac} ~ \sup_{t \in [0, T]} \E^\delta[|U^\delta_t|^{p' + 1}]  =\sup_{\delta \in \Ac} \E^\delta[|U^\delta_T|^{p' + 1}]  < \infty, $$
because $U^\delta$ is a negative $P^\delta$-supermartingale. This leads to \eqref{step5} using H\"older inequality, the uniform boundedness of the intensities of $N^a$ and $N^b$ and that  $e^{- \gamma Y} = U^\delta e^{\gamma \int_0^\cdot ( \delta^a_u dN^a_u + \delta^b_u dN^b_u + Q_u dS_u )}$.\\

\noindent {\it Step 6.} We finally prove uniqueness of the representation. Let $(Y_0,Z),(Y_0',Z')\in\R\times{\cal Z}$ be such that $\xi=Y^{Y_0,Z}_T=Y^{Y'_0,Z'}_T$. By following the line of the verification argument in the proof of Theorem \ref{thm:Agent} (ii), we obtain the equality $Y^{Y_0,Z}_t = Y^{Y'_0,Z'}_t$ by considering the value of the continuation utility of the market maker
 \begin{eqnarray*}
 -\exp(- \gamma Y^{Y_0,Z}_t)
 \;=\; 
 -\exp({- \gamma Y^{Y_0',Z'}_t})
 &=&
 \underset{\delta \in \Ac}{\text{ess sup}} \; 
 \E^\delta_t [-e^{-\gamma (\text{PL}^{\delta}_T -\text{ PL}^{\delta}_t+ \xi)}],
 ~~t\in[0,T].
 \end{eqnarray*}
This in turn implies that $Z^i_t dN^i_t ={Z'}^i_t dN^i_t= d[ Y^{Y_0,Z} ,  N^i]_t$, $i\in\{a,b\}$, and $Z^S_t\sigma^2 dt={Z'}^S_t\sigma^2 dt=d\langle Y,S\rangle_t$, $t\in[0,T]$. Hence $(Y_0,Z)=(Y'_0,Z')$.

\subsection{Proof of Theorem \ref{thm:pbPlatform}}

The proof of the main result of Theorem \ref{thm:pbPlatform} requires the following technical result. We observe that this is the place where the first integrability condition in \eqref{condition:xi} is needed.

\begin{lemma}\label{unifintegr:K}
Let $Z\in \cal Z$. There exists $C>0$ and $\varepsilon>0$ such that
$$\sup_{t\in [0,T]}\mathbb E^{\hat\delta(Y^{\hat Y_0, Z}_T)} [|K_t^Z|^{1+\varepsilon}] \leq C.$$
\end{lemma}

\begin{proof}
We use again the notation $K^Z_t :=
 e^{-\eta ( c(N^a_t+ N^b_t) - Y^{0,Z}_t )},$ $t\in[0,T]$, for all $Z\in \mathcal Z$.
Let $p>1$. By using H\"older's inequality and the uniform boundedness of the intensities of $N^a$ and $N^b$, we deduce that there exists $C'>0$ such that
 \begin{align*}
 \mathbb E^{\hat\delta(Y^{\hat Y_0, Z}_T)}[|K_t^Z|^p]&\leq C' \mathbb E^0[(e^{-\gamma Y_t^{0,Z}})^{-\frac{p'\eta}{\gamma}}]^{\frac p{p'}},
 \end{align*}
 with any $p'>p$. Thus,
  \begin{align*}
 \mathbb E^{\hat\delta(Y^{\hat Y_0, Z}_T)}[|K_t^Z|^p]
 &\leq 
 C' \left (1+ \mathbb E^0\Big[(e^{-\gamma Y_t^{0,Z}})^{-\frac{p'\eta}{\gamma}}\Big]\right)\\
 &=
 C' \left (1+ \mathbb E^0\left[\Big(-\sup_{\delta\in \cal A} 
                                                   \mathbb E^\delta_t\Big[-e^{-\gamma(Y_T^{0, Z}+PL_T^\delta-PL_t^\delta)}
                                                                                 \Big]
                                             \Big)^{-\frac{p'\eta}{\gamma}}\right]\right).
 \end{align*}
 By Jensen's inequality and H\"older's inequality, we deduce for any $p''>p'$ that
 $$
 \mathbb E^{\hat\delta(Y^{\hat Y_0,Z}_T)}[|K_t^Z|^p]
 \leq 
 C' \left (1+ \mathbb E^0\left[\sup_{\delta\in \cal A} \mathbb E^\delta_t[e^{p'\eta(Y_T^{0,Z}+PL_T^\delta-PL_t^\delta)}]\right]\right)
 \leq 
 C' \left (1+ \mathbb E^0\Big[\sup_{\delta\in \cal A} \mathbb E^\delta_t[e^{p''\eta Y_T^{0,Z}}]\Big]\right).
 $$
 By using a dynamic programming principle, similarly to the proof of Lemma \ref{lemme_DPP} by noticing that the family $\left(\widetilde J(\mu,\delta)= \mathbb E^\delta_\tau[e^{p''\eta Y_T^{0,Z}}]\right)_{\mu \in \cal A_\tau}$ is directly upwards, we get
    \begin{align*}
 \mathbb E^{\hat\delta(Y^{\hat Y_0,Z}_T)}[|K_t^Z|^p]
 &\leq C' \left (1+ \sup_{\delta\in \cal A}\mathbb E^\delta \left[e^{p''\eta Y_T^{0, Z}}\right]\right).
 \end{align*}
 By setting $\varepsilon=\frac{\eta'-\eta}3$, if we take $p=1+\varepsilon$, then $p'=p+\varepsilon$ and $p''=p'+\varepsilon$, we obtain 
\begin{align*}
 \mathbb E^{\hat\delta(Y^{\hat Y_0,Z}_T)}[|K_t^Z|^{1+\varepsilon}]
 &\leq C' \left (1+ \sup_{\delta\in \cal A}\mathbb E^\delta \left[e^{\eta' Y_T^{0,Z}}\right]\right).
 \end{align*}
 From the definition of $\mathcal Z$ (involving the first integrability condition in \eqref{condition:xi}), we get
 \begin{eqnarray*}
 \mathbb E^{\hat\delta(Y^{\hat Y_0, Z}_T)}[|K_t^Z|^{1+\varepsilon}]\leq C,
 ~t\in [0,T],
 &\mbox{with}&
 C=C'\left (1+ \sup_{\delta\in \cal A}\mathbb E^\delta \left[e^{\eta' Y_T^{0,Z}}\right]\right)<+\infty.
 \end{eqnarray*}
\end{proof}

\noindent \textbf{Proof of Theorem \ref{thm:pbPlatform}} In order to prove the theorem, we verify that the function $v$ introduced in \eqref{u=etB} coincides at $(0,Q_0)$ with the value function of the reduced exchange problem \eqref{Pb2}, with maximum achieved at the optimal control $\hat Z$. \\

\noindent The function $v$ is negative bounded and has bounded gradient. Moreover, since $\delta_\infty\ge\Delta_\infty$, it follows that $v$ is a solution of the HJB equation \eqref{HJB} of the exchange reduced problem, see Lemma \ref{lemmaAnnex}. For $Z\in{\cal Z}$, denote
 \begin{eqnarray*}
 K^Z_t 
 \;=\;
 e^{-\eta \big( c(N^a_t - N^a_0+ N^b_t - N^b_0) - Y^{0,Z}_t \big)},
 &t\in[0,T].&
 \end{eqnarray*}
By direct application of It\^o's formula, and substitution of $\partial_t v$ from the HJB equation satisfied by $v$, we see that
\begin{eqnarray} 
 d\big[v(t,Q_t)K^Z_t\big]
 &=&
 K^Z_{t-}\Big( (  h^Z_t - {\cal H}_t )dt +\eta v(t,Q_t) Z^s_tdS_t
 \nonumber 
 \\
 && 
 + \sum_{i=a,b} \big[v(t,Q_{t-}+\Delta Q_t)e^{-\eta(c-Z^i_t)} -v(t,Q_{t-}) 
                                                           \big] d\widetilde N^{\hat\delta(Y^{\hat Y_0, Z}),i}_t
              \Big),
 \label{principal_term}
 \end{eqnarray}
where, using the notations of \eqref{HE} and the subsequent equations,
 \begin{eqnarray*}
 {\cal H}_t 
 &:=& 
 H_E\big(Q_t, v(t,Q_{t}), v(t,Q_{t}+1), v(t,Q_{t}-1)\big),
 \\
 h^Z_t 
 &:=&
 h^1_E\big(Q_t, v(t, Q_t), Z^S_t)
 +\1_{\{Q_t>-\bar q\}} h^0_E\big(v(t,Q_{t}),v(t,Q_{t}-1),  Z^a_t\big)\\
 &&
 \hspace{36mm}
 +\1_{\{Q_t<\bar q\}} h^0_E\big(v(t,Q_{t}),v(t,Q_{t}+1),  Z^b_t\big).
 \end{eqnarray*}
Exploiting the fact that $v$ is bounded and that $K^Z$ is uniformly integrable, see Lemma \ref{unifintegr:K}, we get that $\big(v(t,Q_t) K^Z_t\big)_{t \in [0,T]}$ is a $\P^{\hat \delta(Y^{\hat Y_0, Z}_T)}$-supermartingale  and by Doob-Meyer decomposition theorem, the local martingale term in \eqref{principal_term} is a true martingale. Hence
 \begin{eqnarray*}
 v(0,Q_0)
 &=&
 \mathbb{E}^{\hat\delta(Y^{\hat Y_0,Z}_T)}\Big[v(T,Q_T)K^Z_T + \int_0^TK^Z_{t}( {\cal H}_t - h^Z_t)dt \Big]
 \\
 &\ge&
 \mathbb{E}^{\hat\delta(Y^{\hat Y_0,Z}_T)}\big[v(T,Q_T)K^Z_T\big]
 \;=\;
 \mathbb{E}^{\hat\delta(Y^{\hat Y_0,Z}_T)}[-K^Z_T],
 \end{eqnarray*}
by the boundary condition $v(T,.)=-1$. By arbitrariness of $Z\in{\cal Z}$, this provides the inequality $v(0,Q_0)\ge \sup_{Z\in{\cal Z}}\mathbb{E}^{\hat\delta(Y^{\hat Y^0, Z}_T)}[-K^Z_T]=v_0^E$. 

On the other hand, consider the maximizer $\hat Z$ of the reduced exchange problem, induced by the feedback controls $\hat z$ in \eqref{opt-z}. As $\hat Z$ is bounded, it follows that $\hat Z\in\cal Z$. Moreover, $h^{\hat Z} - {\cal H}=0$, by definition, so that the last argument leads to the equality $v(0,Q_0)=\mathbb{E}^{\hat\delta(Y^{\hat Y_0, \hat Z}_T)}\big[-K^{\hat Z}_T\big]$, instead of the inequality. This shows that $v(0,Q_0)=v_0^E$, the reduced exchange problem of \eqref{Pb2}, with optimal control $\hat Z$. From Theorem \ref{thm:Agent}, the corresponding optimal market maker response of the market maker is given by \eqref{optimalSpread} with $\xi = Y^{\hat Y_0, \hat Z}_T$. Moreover, Condition \eqref{cond_infty} implies that $ | - Z^i_t + \frac{1}{\gamma} \log(1 +\frac{\sigma k}{k}) |  \leq \delta_\infty , \quad i = a, b.$
Hence the optimal effort can be reduced to \eqref{optimalOptimalSpread}.

\subsection{Proof of Theorem \ref{thmnash}}\label{appendix:nash}

The proof of Theorem \ref{thmnash} is divided in two steps. First we show that there exists a symmetric and Markovian Nash equilibrium for the problem \eqref{Ppbi}. Then, we show that this Nash equilibrium is unique among the class of symmetric Nash equilibria.\\

\noindent \textbf{Existence of a symmetric Markovian Nash equilibrium.}
We denote by $v$ a smooth solution to \eqref{edp_v^i_hjb} or equivalently \eqref{edp_v^i}. Note that $\zeta^0$ as defined in Theorem \ref{thmnash} is a deterministic function of $t,Q_t$. Denote by $\mathcal K^{\zeta}_t$ the process defined for any $\zeta\in \mathcal Z$ by
\[
\mathcal K_t^{\zeta}
:= 
e^{\eta\big\{ \int_0^T\sum_{i=a,b} (\zeta^i_t-\frac cN) dN^i_t 
                                +\zeta^S_t dS_t 
                                + [ \hat H^S(Q_t,\zeta^{S,0}_t,\zeta^S_t) 
                                      - \sum_{i=a,b} 
                                         \hat H^i(Q_t,\zeta^{i,0}_t,\zeta^i_t)\big) dt    
          \big\}}, 
\]
with
                                    \begin{equation}\label{def:hatHO} \hat H^i(q,\tilde z,z)= H^i(z+(N-1)\tilde z,q) - \frac{N-1}N H^i(N\tilde z,q).  \end{equation}
Note that 
\begin{eqnarray*}
d\mathcal K_t^{\zeta}
&=&
\eta \mathcal K_t^{\zeta} \zeta^S_t dS_t +\sum_{i=a,b}\mathcal K_{t-}^\zeta (e^{\eta (\zeta^i_t-\frac cN) }-1)dN_t^i
\\
&&
+\eta K_t^\zeta \big(  \hat H^S(Q_t,\zeta^{S,0}_t,\zeta^S_t) -\underset{i=a,b}{\sum} \hat H^i(Q_t,\zeta^{i,0}_t,\zeta^i_t)+\frac\eta2 \sigma^2 |\zeta^S_t|^2 \big)dt.
\end{eqnarray*}
Applying Ito's Formula, we obtain
\begin{eqnarray*}
d(v(t,Q_t)\mathcal K_t^{\zeta})
&=&
\mathcal K_t^\zeta \partial_t v(t,Q_t)dt
\\
&&
+\mathcal K_t^\zeta \big(\eta v(t,Q_t) \big(  \hat H^S(Q_t,\zeta^{S,0}(Q_t),\zeta^S_t) -\underset{i=a,b}{\sum} \hat H^i(Q_t,\zeta^{i,0}_t,\zeta^i_t)+\frac\eta2 \sigma^2 |\zeta^S_t|^2 )\big)dt
\\
&&
+\eta \mathcal K_t^{\zeta} v(t,Q_t)\zeta^S_t dS_t +\mathcal K_{t-}^\zeta\sum_{i=a,b} (v(t,Q_t+\Delta Q_t)-v(t,Q_t))e^{\eta (\zeta^i_t-\frac cN) } dN_t^i
\\
&=&
\mathcal K_t^\zeta\big( \partial_t v(t,Q_t)+\eta v(t,Q_t)F_t^\zeta \big)dt+\eta \mathcal K_t^{\zeta} v(t,Q_t)\zeta^S_t dS_t 
\\
&&
+\mathcal K_{t-}^\zeta\sum_{i=a,b} (v(t,Q_t+\Delta Q_t)-v(t,Q_t))e^{\eta (\zeta^i_t-\frac cN) } d\tilde N_t^i,
\end{eqnarray*}
where
\begin{eqnarray*}
F_t^\zeta
&=& 
F^S(t,Q_t,\zeta^S_t)+F^0(t,Q_t,v(t,Q_t),v(t,Q_t+1),\zeta^b_t)\mathbf 1_{Q_t<\overline Q}
\\
&&+F^0(t,Q_t,v(t,Q_t),v(t,Q_t-1),\zeta^a_t)\mathbf 1_{Q_t>-\overline Q}.
\end{eqnarray*}

%with
%
%\[ F_t^S(q,z)=  \hat H_t^S(z)+\frac\eta2 \sigma^2 |z|^2,  \]
%\[F_t^0(q,y,y',z)=  \lambda(\hat\delta(z+\frac{\tilde Z_t}2))\big(\frac{y'}y\frac{e^{\eta (z-\frac c2) }}{\eta}-\frac1\eta- \frac{\sigma}{k+\sigma\gamma}\big) +  \lambda(\hat\delta(\tilde Z_t)) \frac{\sigma}{k+\sigma\gamma}. \]
\noindent Since $v$ satisfies HJB equation \eqref{edp_v^i_hjb}, we deduce that $\mathbb E^{\hat {\delta}(\zeta+(N-1)\zeta^0)}[-\mathcal K_T^\zeta]\leq v(0,Q_0)$,
with equality for $\zeta^S=\zeta^{S,0}$ and $\zeta^i=\zeta^{i,0}$.\\

\noindent \textbf{Uniqueness among the set of general symmetric Nash equilibria.} We now prove that if there exists a symmetric Nash equilibrium, it is unique and given by the Markovian equilibrium $\xi^0$ defined by \eqref{optimal:contract:2ex}. Let $\xi^0$ characterized by \eqref{eq:nashsym} for general $\zeta^0$. We consider the dynamic value function of any exchange given $\xi^0$ fixed by the other, denoted by $V_t(\xi^0)$ and defined in view of Remark \ref{remark:symmetric} by
\[ V_t(\xi^0)=
 e^{\eta\hat Y_0} \underset{\zeta \in  \mathcal Z}{\text{ess sup}}\, \E_t^{\hat\delta(\zeta+(N-1) \zeta^0)} 
                                     \Big[-e^{-\eta\big (\int_t^T\underset{i=a,b}{\sum}(\frac cN-\zeta^i_r) dN_r^i -\int_t^T \zeta_r^SdS_r -\int_t^T \hat H^S_r( \zeta^S_r)- \underset{i=a,b}{\sum} \hat H_r^i(\zeta_r^i) dr \big)}\Big],
\]
with $\hat H_r^S(z)= \hat H_r^S(Q_r, \zeta^{S,0}_r,z)\; \text{ and }\; \hat H_r^i(z)= \hat H_r^i(Q_r, \zeta^{i,0}_r,z).$ By using a DPP similarly to \ref{lemme_DPP}, we prove that $(V_t(\xi^0)\mathcal K_t^\zeta)_{t\in [0,T]}$ is a $\mathbb P^{\zeta+ (N-1)\zeta^0}$-super martingale. The martingale condition thus provides the optimal $\zeta$ played by the representative exchange given that the others choose $\zeta^0$. By using a Doob-Meyer decomposition, the martingale property leads to the solution of the following BSDE
\begin{equation}\label{bsde:R} dR_t= U^S_tdS_t+\sum_{i=a,b} U^i_td\widetilde{N^i_t } - F_t(U_t,Q_t)dt,\; R_T=0\end{equation}
with
\begin{align*}
 F_t(u,q)&:=\sup_{\zeta^S} \big( -\hat H_t^S(\zeta^S)-\frac{\sigma^2\eta}2 |\zeta^S-u^S|^2) \big)\\
 & +\sum_{i=a,b}\sup_{\zeta^i}\big(  \hat H_t^i(\zeta^i) - \lambda_t^{i,\zeta} \big( u^i+ \frac{1-e^{-\eta(u^i-\zeta^i+\frac cN)}}{\eta}\big) \big), \end{align*}
with $\lambda_t^{i,\zeta} =\lambda(\hat \delta(\zeta^i+(N-1)\zeta^{i,0}))$.We directly derive the maximizers 
\[ \zeta^{S,\star}_t=-\frac{\gamma (N-1)}{\eta+\gamma} \zeta_t^{S,0}-\frac{\gamma}{\eta+\gamma}Q_t-\frac{\eta}{\eta+\gamma} U_t^S,\]

\[ \zeta^{i,\star}_t= U^i_t +\frac cN +\frac 1\eta \log\big(\frac{k(k+\sigma(\gamma-\eta))}{(k+\sigma\eta)(k+\sigma\gamma)}+\frac{U^i_tk\eta}{k+\sigma\eta}\big).\]

\noindent Since $\zeta^0$ is assumed to be a symmetric Nash equilibrium, we obtain from Definition \ref{def:nash:sym} that $\zeta^{S,0}_t$ and $\zeta^{i,0}_t$ are necessarily uniquely determined as function of $Q_t$ and $U_t$ by

\[\zeta^{S,0}_t= -\frac{\gamma}{\eta+N\gamma}Q_t-\frac{\eta}{\eta+N\gamma} U_t^S,  \]
and
\[ \zeta^{i,0}_t= U^i_t +\frac cN +\frac 1\eta \log\big(\frac{k(k+\sigma(\gamma-\eta))}{(k+\sigma\eta)(k+\sigma\gamma)}+\frac{U^i_tk\eta}{k+\sigma\eta}\big).\] Hence, we note that the BSDE \eqref{bsde:R} is Markovian. The integro-partial differential equation associated with this BSDE remains to solve \eqref{edp_v^i_hjb} for which we know that there exists a continuous solution given by $v(t,q)$. We thus deduce that if there is a symmetric Nash equilibrium, it is Markovian in the sense of Definition \ref{def:markovnash} and the first step of the proof shows that it is unique.

\subsection{First best exchange problem}\label{sec:firstbest}

In this section, we analyze the first best problem of the exchange. In  this setting the exchange optimally chooses the contract $\xi$ and the optimal bid-ask posting policy of the market maker under her participation constraint. Introducing a Lagrange multiplier $\lambda>0$ to penalize for this constraint, we  reduced the first best exchange value function to the unconstrained problem:
 \[
 V_0^{FB}
 =
 \inf_{\lambda>0}\sup_{\xi \in \mathfrak C,\delta\in \mathcal A}
 \mathbb E^{\mathbb P^\delta}
 \big[ -e^{-\eta(c(N_T^a-N_T^b)-\xi)}-\lambda e^{-\gamma(\xi+X_T+Q_TS_T)}-\lambda R\big],
 \]
with $\mathfrak C=
 \Big\{ \xi,\, \text{ ${\cal F}_T$-measurable such that \eqref{condition:xi} is satisfied} 
 \Big\}$.
 
 We first compute the supremum on $\xi$ by fixing $\lambda,\delta$. The first order condition in $\xi$ is $e^{-\eta(c(N_T^a-N_T^b)-\xi^\star_\lambda)}=\frac{\lambda\gamma}{\eta}e^{-\gamma (\xi^\star_\lambda+X_T+Q_TS_T)} $, implying
 \[ 
 \xi^{\star}_\lambda
 = 
 \frac{1}{\eta+\gamma}\big(\log(\frac{\lambda\gamma}\eta)-\gamma (X_T+Q_TS_T)+\eta c (N_T^a+N_T^b) \big).
 \]
Substituting this expression, we see that
\begin{align*}
V_0^{FB}
&
=
\inf_{\lambda>0}\sup_{\delta\in \mathcal A}
\mathbb E^{\mathbb P^\delta}\big[ -\lambda\frac{\eta+\gamma}{\eta}
                                                       e^{-\gamma (\xi^\star_\lambda+X_T+Q_TS_T)}
                                                       -\lambda R
                                                \big]\\
&
=
\inf_{\lambda>0}\sup_{\delta\in \mathcal A}
\mathbb E^{\mathbb P^\delta}
\big[ -\lambda\frac{\eta+\gamma}{\eta} 
        \big(\frac{\eta}{\lambda\gamma} \big)^{\frac{\gamma}{\eta+\gamma}}
        e^{-\frac{\gamma\eta}{\gamma+\eta} (X_T+Q_TS_T+c(N_T^a+N_T^b))}
        -\lambda R
\big]\\
&=\inf_{\lambda>0} \lambda\big[\frac{\eta+\gamma}{\eta} \big(\frac{\eta}{\lambda\gamma} \big)^{\frac{\gamma}{\eta+\gamma}} \widetilde{V_0}-R\big],
~\mbox{with}~
\widetilde{V}_0
= 
\sup_{\delta\in \mathcal A}
\mathbb E^{\mathbb P^\delta}
\big[ -e^{-\frac{\gamma\eta}{\gamma+\eta} (X_T+Q_TS_T+c(N_T^a+N_T^b))}\big].
\end{align*}
As $\widetilde{V}_0$ is independent of $\lambda$, we obtain the optimal Lagrange multiplier $\lambda^\star
=
\frac{\eta}{\gamma}\Big(\frac{\widetilde{V}_0}{R}\Big)^{1+\frac{\eta}\gamma}$, and we deduce the optimal first best contract:
\[ 
\xi^{\star}
=
\xi^\star_{\lambda^*}
= 
\frac{1}{\eta+\gamma}\Big(\log(\frac{\lambda^{\star}\gamma}\eta)-\gamma (X_T+Q_TS_T)+\eta c (N_T^a+N_T^b) \Big).
\]
We finally solve the problem $\widetilde{V_0}$. Note that by setting $\tilde \delta:=\delta+c$ in view of the definition of $\mathbb P^\delta$ given by \eqref{proba_change} together with \eqref{intensity} we get
\begin{eqnarray*}
\widetilde{V}_0
= 
\sup_{\tilde \delta\in \tilde{\mathcal A}}
\mathbb E^{\mathbb P^{\tilde \delta-c}}
\big[ -e^{-\Gamma(X_T+Q_TS_T)}\big],
&\mbox{with}&
\Gamma=\frac{\gamma\eta}{\gamma+\eta},
\end{eqnarray*}
and where $\tilde {\mathcal A}$ is defined similarly to $\mathcal A$ with bound $\delta_\infty+c$. We are then reduced to the framework of \cite{avellaneda2008high,gueant2013dealing} so that the optimal bid-ask spreads are given by
 $$ 
 \widetilde{\delta}^{i}_t
 = 
 -c
 + \frac{1}{\Gamma}\log(1 +\frac{\sigma \Gamma}{k} )
 +\frac{\sigma}{k} \log\Big( \frac{\widetilde{u}^{FB}(t,Q_{t-})}{\widetilde{u}^{FB}(t,Q_{t-}+\eps_i)}\Big), 
 ~~i\in\{b,a\},
 ~(\eps_a,\varepsilon_b)=(-1,1).
 $$
where $ \widetilde{u}^{FB}$ is the unique solution of the linear differential equation 
\begin{eqnarray}\label{FB}
\widetilde{u}^{FB}\big|_{t=T}=1,
~~
\partial_t \widetilde{u}^{FB}(t,q) - F_{C_1^{FB}, \widetilde{C_1}^{FB}}(q,\widetilde{u}^{FB}(t,q), \widetilde{u}^{FB}(t,q+1), \widetilde{u}^{FB}(t,q-1))  = 0,
\end{eqnarray}
with constants $C_1^{FB} = \frac{\sigma \Gamma k }2$ and $\widetilde{C_1}^{FB} = A \big(1+\frac{\sigma \Gamma}{k}\big)^{-(1+\frac{\sigma \Gamma}k)} $, and so that 
\[v^{FB}(0,0)=\widetilde{V_0}^{FB}, \text{ with } \widetilde{u}^{FB}=(-\widetilde{v}^{FB})^{-\frac{k}{\sigma\eta}}.\]

Since the solution of PDE \eqref{FB} is different from the solution of \eqref{linearPDE}, we deduce that the value function of the exchange in the first best case does not coincide with his value function in the second best model. \\

\noindent {\bf \Large Acknowledgements}\quad
We are grateful to Ang\'elique Begrand, Luxi Chen and Laurent Fournier from Euronext for helpful comments.

{\small
\bibliographystyle{abbrv}
\bibliography{BibEEMRT}
}
\end{document}